\documentclass[11pt]{article}

\usepackage[english]{babel}
\usepackage[latin1]{inputenc}
\usepackage[T1]{fontenc}
\usepackage[usenames]{color}
\usepackage{amsmath}
\usepackage{amssymb}
\usepackage{hhline}
\usepackage{multirow}
\usepackage{colortbl}
\usepackage{enumerate}
\usepackage{array}
\usepackage{soul}
\usepackage{graphicx}
\usepackage{eurosym}
\usepackage{mathrsfs}
\usepackage{dsfont}
\usepackage{amsthm}
\usepackage{makecell}
\usepackage[titletoc]{appendix}
\usepackage[titles]{tocloft}
\usepackage{indentfirst}

\newtheorem{proposition}{Proposition}[section]
\newtheorem{remark}{Remark}[section]

\title{A spectral method for an Optimal Investment problem with Transaction Costs under Potential Utility\thanks{Research supported by Spanish MINECO
under grants MTM2013-42538-P and MTM2016-78995-P. The first author acknowledges the support of European Cooperation in Science and Technology through COST Action IS1104.}}

\author{Javier de Frutos\thanks{Instituto de Matem\'{a}ticas (IMUVA), Universidad de Valladolid, Paseo de Bel\'{e}n 7, Valladolid, Spain. e-mail:frutos@mac.uva.es} and V\'{\i}ctor Gat\'{o}n\thanks{Instituto de Matem\'{a}ticas (IMUVA), Universidad de Valladolid,  Paseo de Bel\'{e}n 7, Valladolid, Spain. e-mail:vgaton@mac.uva.es}}

\begin{document}

\maketitle

\begin{abstract}This paper concerns the numerical solution of the finite-horizon Optimal Investment problem with transaction costs under Potential Utility. The problem is initially posed in terms of an evolutive HJB equation with gradient constraints. In \cite{DaiYi}, the problem is reformulated as a non-linear parabolic double obstacle problem posed in one spatial variable and defined in an unbounded domain where several explicit properties and formulas are obtained. The restatement of the problem in polar coordinates allows to pose the problem in one spatial variable in a finite domain, avoiding some of the technical difficulties of the numerical solution of the previous statement of the problem. If high precision is required, the spectral numerical method proposed becomes more efficient than simpler methods as finite differences for example.

\end{abstract}

\section{Introduction}

This paper concerns with the numerical solution of the finite horizon optimal investment problem with transaction costs under Potential Utility. Let us consider an investor whose wealth can be inverted in a risky stock and in a riskless bank account. We suppose that the investor is risk averse with constant relative risk aversion (CRRA). In \cite{Merton}, Merton showed that, in absence of transaction costs, the problem can be explicitly solved. The optimal strategy consists in keeping a fixed proportion between the money invested in the risky asset and the bank account. When transaction cost are considered, the Merton strategy is unfeasible because it requires a continuous portfolio rebalancing with unbounded costs.

Proportional transaction costs were first introduced in \cite{Magill}.
More recently, in \cite{DaiYi}, the problem was reformulated as a non-linear parabolic double obstacle problem posed in one spatial variable, defined in an unbounded domain. Several explicit properties and formulae were obtained in \cite{DaiYi}, although explicit formulae for the solution are not available. The problem was numerically solved in \cite{Arregui}, where the authors employ a characteristics method with a projected relaxation scheme. The scheme proportionates satisfactory results with good agrement with the results in \cite{DaiYi}.

When we want to solve financial problems, like finding investment strategies or pricing derivative contracts, in general, there is no known closed form solution of the different problems and several numerical methods have been employed. Without the aim to be exhaustive, Monte-Carlo based methods (\cite{Duan2}, \cite{Stentoft}), Piecewise linear interpolations (\cite{Benameur}), Lattice methods (\cite{Lyuu}, \cite{Ritchen}), Finite Elements (\cite{Achdou}) or Spectral methods (\cite{Frutos}) are some of them. A general review of financial problems or models, numerical techniques and software tools can be found in \cite{DuanHandbook}.

The objective of this paper is to construct a spectral method specifically adapted to the Optimal Investment problem with Potential Utility when proportional transaction costs are present. As it is well known, spectral methods  \cite{Canuto} are a class of spatial discretizations for partial differential equations which offer fast convergence in the case of smooth solutions. They are not widely used yet in numerical finances because it is usually believed that the lack of smoothness present in most interesting problems makes spectral methods uncompetitive. However, several papers have used spectral methods for problems in Finance with good results. For instance,  in  \cite{Chiarella} a Fourier-Hermite procedure to the valuation of american options has been presented. In \cite{Frutos} a spectral method based on Laguerre polynomials has been employed to numerical valuation of bonds with embedded options. A Fourier spectral method to compute fast and accurate prices of american options written on asset following GARCH models has been presented in \cite{Breton}. In \cite{Breton2} the authors use an adaptive method with Chebyshev polynomials coupled with a dynamic programming procedure for contracts with early exercise features. In \cite{Oosterlee} a very efficient procedure for asian options defined on arithmetic averages has been proposed. In all cases, the spectral-based methods have been proved to be competitive with other alternatives in terms of precision versus computing time needed to compute the numerical solution.

In the present paper, we restate the problem using polar coordinates. This allows to consider a double parabolic obstacle problem in one spatial-like variable defined in a bounded domain. Furthermore, this formulation avoids the emergence of nonlinear terms simplifying the numerical treatment. We present a Chebyshev spectral approach based on adaptive meshes to locate the optimal frontiers.
Although some of the numerical difficulties that appear with the parabolic double obstacle problem are avoided with our approach, we still have to deal with the so-called Gibbs effect, which comes from the fact that the objective function is continuous but not differentiable at maturity. We show that this issue can be circumvented by using a time-adapted spatial mesh. We show that our approach is efficient by comparing it with a standard finite difference scheme.

The outline of the paper is as follows. In Section \ref{Ch3TOIPrecon}, a description of the Optimal Investment problem as it can be found in \cite{Davis} or \cite{Shreve} is presented. In Section \ref{Ch3DPR}, the problem is reformulated as a parabolic double obstacle problem as it was done in \cite{DaiYi}. Afterwards, we propose an equivalent formulation of the problem employing polar coordinates. Section \ref{Ch3NM} is devoted to a mesh-adapted Chebyshev-collocation method which solves the problem of Section \ref{Ch3DPR}. In Section \ref{Ch3numresults} we perform the numerical analysis of the method. Section \ref{Ch3Conclus} presents some conclusions and future research.

\section{The Optimal Investment Problem}\label{Ch3TOIPrecon}

We consider an optimal investment problem with transaction costs, \cite{Davis}, \cite{Shreve}.
Let $\left(\Omega,\mathscr{F},P\right)$ be a filtered probability space. Let us consider an investor who holds amounts $X(t)$ and $Y(t)$ in a bank and a stock account respectively. The dynamics of the processes is
\begin{equation}\label{Ch3ecudinacc}
\begin{aligned}
dX(t) &=rX(t)dt-(1+\lambda)dL(t)+(1-\mu)dM(t), && X(t_0)=x, \\
dY(t) &=\alpha Y(t) dt + \sigma Y(t) dz_t+dL(t)-dM(t), && Y(t_0)=y,
\end{aligned}
\end{equation}
where $r$ denotes the constant risk-free rate, $\alpha$ is the constant expected rate of return of the stock, $\sigma>0$ is the constant volatility of the stock and $z_t$ is a standard Brownian motion such that $\mathscr{F}^{z}_t\subseteq\mathscr{F}$ where $\mathscr{F}^{z}_t$ is the natural filtration induced by $z_t$. We suppose that $L(t)$ and $M(t)$ are adapted, right-continuous, nonnegative and nondecreasing processes representing the cumulative monetary values of the stock purchased or sold respectively and $\lambda\ge 0$ and $0\le\mu<1$, represent the constant proportional transaction costs incurred on the purchase or sale of the stock. In this paper we assume that $\lambda+\mu>0$.

The finance meaning of equation (\ref{Ch3ecudinacc}) is natural. Along time, the rate of change of the amount of money invested in the risky asset, represented by the stochastic process $Y(t)$, evolves according to a standard geometric brownian motion modified by the difference between the amount of money invested in buying stock, $dL(t)$, and the amount of money obtained selling stock, $dM(t)$. At the same time, the value of the bank account, $X(t)$, is instantaneously increased by the difference $-(1+\lambda)dL(t)+(1-\mu)dM(t)$, that represents the net flow of money resulting from stock negotiations, including the transaction costs. Processes $L(t)$ and $M(t)$ can be financially understood as an historical record of the total purchases and sales of stock of the investor.

The net wealth is the money the investor would have if he closes his positions. It can be written as
\begin{equation}\label{Ch3ecunewwealth_long}
W(t)=
 X(t)+(1-\mu)Y(t), \quad \text{if} \ Y(t)\geq0,
\end{equation}
if the investor is long in the stock or
\begin{equation}\label{Ch3ecunewwealth_short}
W(t)= X(t)+(1+\lambda)Y(t), \quad \text{if} \ Y(t)<0,
\end{equation}
in case the investor is short in the stock.

Let $U(w)$ be a utility function, that is, a continuous, strictly increasing, concave function. The optimal value function is given by:
\begin{equation}\label{Ch3defvarpphi}
\varphi(x,y,t)=\sup_{(L,M)\in A_t(x,y)}E\left[\left.U\left(W(T)\right)\right|(X(t),Y(t))=(x,y)\right],
\end{equation}
for all $ (x,y,t)\in\mathbb{S}\times[0,T]$, where $A_t(x,y)$ is the set of admissible strategies, defined as the set of processes $(L,M)$ such that
if $(X(t),Y(t))=(x,y)\in\mathbb{S}$ then  $(X(\tau),Y(\tau))\in\mathbb{S}$ for $t\le \tau\le T$ and where $\mathbb{S}$ is
the Solvency Region,
\begin{equation}\label{Ch3defsolvreg}
\mathbb{S}=\left\{(x,y)\in \mathbb{R}^2  \mid x+(1+\lambda)y >0,  x+(1-\mu)y>0\right\}.
\end{equation}

In this paper, we assume that $U(w)$ is a potential function (constant relative risk aversion utility function) of the form
\begin{equation*}
U\left(w\right)=\frac{w^{\gamma}}{\gamma},
\end{equation*}
for some constant $\gamma$, $0<\gamma<1$.

 The optimal value function (\ref{Ch3defvarpphi}), see \cite{Shreve}, is the viscosity solution in $\mathbb{S}\times [0,T]$ of
\begin{equation}\label{Chp3HJBecu1}
\min\left\{-\varphi_t-{\mathcal{L}}\varphi,\ -(1-\mu)\varphi_x+\varphi_y, \ (1+\lambda)\varphi_x-\varphi_y \right\}=0,
\end{equation}
subject to:
\begin{equation}\label{Chp3HJBecu2}
\varphi(x,y,T)=
\left\{
\begin{aligned}
& U(x+(1-\mu)y), \quad \text{if} \ y>0, \\
& U(x+(1+\lambda)y), \quad \text{if} \ y\leq0,
\end{aligned}
\right.
\end{equation}
where
\begin{equation}\label{Chp3HJBecu3}
\mathcal{L}\varphi=\frac{1}{2}\sigma^2y^2\varphi_{yy}+\alpha y \varphi_y+rx\varphi_x.
\end{equation}

The existence and uniqueness of a viscosity solution of (\ref{Chp3HJBecu1})-(\ref{Chp3HJBecu2}) has been proved in \cite{Davis}. There, it is proved that at any time $t$, the spatial domain is divided in three regions, namely, in financial terms, the Buying Region $\text{BR}(t)=\{ (x,y)|(1+\lambda)\varphi_x-\varphi_y=0\}$,
the Selling Region $\text{SR}(t)= \{(x,y)|-(1-\mu)\varphi_x+\varphi_y=0\}$ and the No Transactions Region
 $\text{NT}(t)=\{ (x,y)| -\varphi_t-\mathcal{L}\varphi=0\}$.
 The Selling and Buying Regions do not intersect.

For simplicity in the exposition, we suppose that $\alpha>r$. With this hypothesis, short-selling is always a suboptimal strategy \cite{Cvitanic}, \cite{Merton}, \cite{Shreve}. This means that the optimal trading strategy is always to have a nonnegative amount of money invested in the stock.

\section{Reformulation of the problem.}\label{Ch3DPR}

As remarked in \cite{Davis}, the choice of the Potential Utility function is interesting since it leads to the homothetic property in the optimal value function,
\begin{equation}\label{homothetic}
\varphi(\rho x, \rho y,t)=\rho^{\gamma} \varphi (x,y,t), \quad \rho>0.
\end{equation}
This property is used in \cite{DaiYi} to reduce the dimensionality of the problem. Setting $z=\frac{x}{y}, \ z \in \Omega=(-(1-\mu), \ \infty)$, a new function $G(x,t)=\varphi(x,1,t)$ is introduced in \cite{DaiYi}, so that:
\begin{equation}\label{Checamborig}
\varphi(x,y,t)=y^{\gamma}G\left(z,t\right), \quad w=\frac{1}{\gamma}\log(\gamma G), \quad v(z,t)=w_z(z,t).
\end{equation}

In \cite{DaiYi}, the authors prove that $v(z,t)$ is the solution of an one dimensional parabolic double obstacle problem with two free boundaries equivalent to  (\ref{Chp3HJBecu1}).

Furthermore, it is also proved in \cite{DaiYi}, that there exist two continuous monotonically increasing functions
\begin{equation}\label{Ch3FunfronC}
\text{BR}^c_F,\text{SR}^c_F:[0,T]\rightarrow(-(1-\mu),+\infty],
\end{equation}
such that $\text{BR}^c_F(t)>\text{SR}^c_F(t), \ \forall t\geq0$. The Buying and Selling Regions are characterized by
\begin{equation*}
\begin{aligned}
\text{\textbf{SR}} &= \left\{(z,t)\in\Omega\times[0,T] \ \mid z\leq\text{SR}^c_F(t), t\in[0, \ T] \right\}, \\
\text{\textbf{BR}} &= \left\{(z,t)\in\Omega\times[0,T] \ \mid z\geq\text{BR}^c_F(t), t\in[0, \ T] \right\}. \\
\end{aligned}
\end{equation*}

Although other properties and explicit formulas are obtained in \cite{DaiYi}, a complete analytical solution is still missing and numerical procedures have to be used, see, for example, \cite{Arregui}.

Here, inspired by \cite{DaiYi}, we take advantage of (\ref{homothetic}) by working in polar coordinates $x=b \cos(\theta)$, $ y=b \sin(\theta)$. It is not difficult to show that  (\ref{Chp3HJBecu1})-(\ref{Chp3HJBecu3}) are equivalent to
\begin{equation}\label{Ch3ecuinicpol}
\min  \left\{ -\varphi_t-\mathcal{L}\varphi,\ -(1-\mu)\mathcal{L}_1\varphi+ \mathcal{L}_2\varphi ,
 (1+\lambda)\mathcal{L}_1\varphi
-\mathcal{L}_2\varphi \right\}=0,
\end{equation}
subject to :
\begin{equation}\label{finalcondition1}
\varphi(b,\theta,T)=\begin{cases}
U(b\cos(\theta)+(1-\mu)b\sin(\theta)), \quad \text{if} \ \theta>0, \\
U(b\cos(\theta)+(1+\lambda)b \sin(\theta)), \quad \text{if} \ \theta\le 0,
\end{cases}
\end{equation}
where
 $$\mathcal{L}_1\varphi=\cos(\theta)\varphi_b-\frac{\sin(\theta)}{b}\varphi_{\theta},\quad
 \mathcal{L}_2\varphi=\sin(\theta)\varphi_b+\frac{\cos(\theta)}{b}\varphi_{\theta} $$
 and
\begin{equation*}
\begin{aligned}
\mathcal{L}\varphi=& \frac{1}{2}\sigma^2\bigl (b\sin(\theta)\bigr)^2
\bigl(\sin^2(\theta)\varphi_{bb}+\frac{2\sin(\theta)\cos(\theta)}{b}\varphi_{b\theta}\\
&\phantom{\frac{1}{2}\sigma^2\bigl(b\sin(\theta)\bigr)^2}
+\frac{\cos^2(\theta)}{b^2}\varphi_{\theta\theta}
+\frac{\cos^2(\theta)}{b}\varphi_b-\frac{2\sin(\theta)\cos(\theta)}{b^2}\varphi_{\theta}\bigr) \\
& +\alpha b \sin(\theta)\bigl(\sin(\theta)\varphi_b+\frac{\cos(\theta)}{b}\varphi_{\theta}\bigr)
+rb \cos(\theta)\bigl(\cos(\theta)\varphi_b-\frac{\sin(\theta)}{b}\varphi_{\theta}\bigr).
\end{aligned}
\end{equation*}

Based on (\ref{homothetic}), we conjecture  a solution to (\ref{Ch3ecuinicpol})  of the form:
\begin{equation*}
\varphi(b,\theta,t)=b^{\gamma} V(\theta,t).
\end{equation*}
Taking into account that
\begin{equation*}
\begin{aligned}
\varphi_b &=\gamma b^{\gamma-1} V,&\varphi_{bb}&=\gamma(\gamma-1)b^{\gamma-2}V,&\varphi_{b\theta}&=\gamma b^{\gamma-1}V_{\theta},& \\
\varphi_{\theta} &=b^{\gamma} V_{\theta},&\varphi_{\theta \theta}&=b^{\gamma}V_{\theta \theta}, \quad &\varphi_t&=b^{\gamma} V_{t},& \\
\end{aligned}
\end{equation*}
and substituting in (\ref{Ch3ecuinicpol}), (\ref{finalcondition1}), we see that $V(\theta,t)$ satisfies,
\begin{align}\label{Ch3enunprobpolar}
\min & \left\{-V_t-g_2(\theta)V_{\theta \theta}-g_1(\theta)V_{\theta}-g_0(\theta) V,
 -V_{\theta}+ \gamma\frac{(1+\lambda)\cos(\theta)-\sin(\theta)}{(1+\lambda)\sin(\theta)+\cos(\theta)}V, \right. \nonumber\\
& \left. \ V_{\theta}- \gamma\frac{(1-\mu)\cos(\theta)-\sin(\theta)}{(1-\mu)\sin(\theta)+\cos(\theta)}V \right\}=0, \  \theta \in(\beta_1, \beta_2),\ t\in[0,T).
\end{align}
subject to:
\begin{equation}\label{Ch3enunprobpolar2}
V(\theta,T)=
\begin{cases}
\frac{1}{\gamma}\left(\cos(\theta)+(1-\mu)\sin(\theta)\right)^{\gamma}, \quad \text{if} \ \theta>0, \\
\frac{1}{\gamma}\left(\cos(\theta)+(1+\lambda)\sin(\theta)\right)^{\gamma}, \quad \text{if} \ \theta\le 0.
\end{cases}
\end{equation}
The functions $g_i$, $i=0,1,2$, are given by
\begin{equation*}
\begin{aligned}
g_0(\theta) &= \gamma \Bigl(\bigl(\frac{1}{2}\sigma^2\sin^2(\theta)(\gamma-1)\sin^2(\theta)+\cos^2(\theta)\bigr)+\alpha\sin^2(\theta)+r\cos^2(\theta)  \Bigr),\\
g_1(\theta) &= (\gamma-1)\sigma^2\cos(\theta)\sin^3(\theta)+(\alpha-r)\sin(\theta)\cos(\theta), \\
g_2(\theta) &= \frac{1}{2}\sigma^2\sin^2(\theta)\cos^2(\theta). \\
\end{aligned}
\end{equation*}

The Solvency Region in the new coordinates is given by:
\begin{equation}\label{Ch3changpolcoord}
b \in [0,\ \infty), \quad \theta \in(\beta_1, \ \beta_2)
\end{equation}
where
\begin{equation}
\beta_1 = \arctan\left(\frac{-1}{1+\lambda}\right), \quad \beta_2 = \arctan\left(\frac{-1}{1-\mu}\right)+\pi. \\
\end{equation}

This formulation has several advantages over the formulation of \cite{DaiYi}. As in \cite{DaiYi}, the problem is one dimensional ((\ref{Ch3enunprobpolar})-(\ref{Ch3enunprobpolar2}) do not depend of $b$), but in our case the domain is bounded ($\theta\in(\beta_1, \ \beta_2)$) . Furthermore, the operators involved in (\ref{Ch3enunprobpolar}) are linear in $V$, whereas in \cite{DaiYi}, the equations contains a nonlinear term.

Next, we characterize the buying and selling regions in terms of the polar coordinates.
First, let us observe that
\begin{equation}\label{Ch3formularelacionpolarorig}
\begin{aligned}
v(z,t) &=-\left(\frac{V_{\theta}(\theta,t)\sin^2(\theta)-\gamma\sin(\theta)\cos(\theta)V(\theta,t)}{\gamma V(\theta,t)}\right),  \\
z &=\cot(\theta),
\end{aligned}
\end{equation}
where $v(z,t)$ is the function defined in (\ref{Checamborig}).

Let us define the functions $\text{BR}_F$ and $\text{SR}_F$ by
\begin{equation*}
\text{BR}^c_F(t)=\cot\left(\text{BR}_F(t)\right), \quad \text{SR}^c_F(t)=\cot\left(\text{SR}_F(t)\right), \quad t\in [0,T],
\end{equation*}
where $\text{BR}^c_F$ and $\text{SR}^c_F$ are the boundaries of the buying and selling regions in cartesian coordinates defined in (\ref{Ch3FunfronC}). The following proposition is an immediate consequence of the results in \cite{DaiYi}.
\begin{proposition}\label{Ch3tradpropertapolar}
Functions $\text{SR}_F$, $\text{BR}_F$ are monotonically decreasing functions.

It holds that $\text{BR}_F(t)<\text{SR}_F(t)$ and that
\begin{equation*}
\text{BR}_F(t)=0, \quad t\in[\hat{t}_0, \ T],\quad \hat{t}_0=T-\frac{1}{\alpha-r}\log\frac{1+\lambda}{1-\mu}.
\end{equation*}

If $\alpha-r-(1-\gamma)\sigma^2>0$, then $\text{BR}_F(\hat{t}_1)=\frac{\pi}{2}$, with
$$\hat{t}_1=T-\frac{1}{\alpha-r-(1-\gamma)\sigma^2}\log\frac{1+\lambda}{1-\mu}.$$

It holds that $\underset{t\rightarrow T}{\lim} \cot\left(\text{SR}_F(t)\right)=(1-\mu)x_{M}$, where
$$x_M=-\frac{\alpha-r-(1-\gamma)\sigma^2}{\alpha-r}$$
is the  Merton line.

If $T\rightarrow\infty$, there exist two values $\text{BR}_s, \text{SR}_s \in (\beta_1, \beta_2)$, such that
\begin{equation*}
\begin{aligned}
\underset{t\rightarrow 0^{+}}{\lim} \text{BR}_F(t) &= \text{BR}_s, \\
\underset{t\rightarrow 0^{+}}{\lim} \text{SR}_F(t) &= \text{SR}_s.
\end{aligned}
\end{equation*}
 The limit values $\text{BR}_s$ and $\text{SR}_s$ are defined by
 $$\cot(\text{BR}_s)=-\frac{a}{a+\frac{k}{k-1}}(1+\lambda),\quad \cot(\text{SR}_s)=-\frac{a}{a+k}(1-\mu),$$
where $a$ and $k$ are the constants defined in \cite[Theorem 6.1]{DaiYi}.

The functions $\text{SR}_F$, $\text{BR}_F$ satisfy (see also \cite[Proposition 3.4.2]{Gaton}),
\begin{equation*}
\beta_1<0\leq \text{BR}(t) \leq \text{SR}(t) \leq  \text{SR}_s<\beta_2, \quad t\in[0,T].
\end{equation*}
\end{proposition}

It is now easy to see that, for $t\in [0,T]$, the buying, selling and no transaction region can be described by as follows:

\noindent \textbf{1.} The buying region is defined by $ \text{BR}=(\beta_1,\ \text{BR}_F(t)]$. In $\text{BR}$ the value function satisfies
\begin{equation}\label{Chp3BRexpformula}
V_{\theta}= \gamma\frac{(1+\lambda)\cos(\theta)-\sin(\theta)}{(1+\lambda)\sin(\theta)+\cos(\theta)}V,
\end{equation}
\noindent \textbf{2.} The Selling region is defined by $\text{SR}=[\text{SR}_F(t), \beta_2)$. In $\text{SR}$ the value function satisfies
\begin{equation}\label{Chp3SRexpformula}
V_{\theta}= \gamma\frac{(1-\mu)\cos(\theta)-\sin(\theta)}{(1-\mu)\sin(\theta)+\cos(\theta)}V.
\end{equation}
\noindent \textbf{3.} The No Transaction Region is defined by $\text{NT}=(\text{BR}_F(t), \text{SR}_F(t))$.
In $\text{NT}$, $V$ satisfies of the following partial differential equations
\begin{equation}\label{Notranstaction}
V_t+g_2(\theta)V_{\theta \theta}+g_1(\theta)V_{\theta}+g_0(\theta) V=0.
\end{equation}

We remark that if the buying ($\text{BR}_F(t)$) and Selling ($\text{SR}_F(t)$) frontiers are known, we can compute the value function $V(\theta,t)$ in $\text{BR}$ and $\text{SR}$ explicitly by a simple integration of equations (\ref{Chp3BRexpformula}) and (\ref{Chp3SRexpformula}) respectively. For $\beta_1<\theta<\text{BR}_F(t)$, we have
\begin{equation}\label{Ch3exfsrbr1}
V(\theta,t)=V(\text{BR}_F(t),t)
\left(\frac{(1+\lambda)\sin(\theta)+\cos(\theta)}{(1+\lambda)\sin(\text{BR}_F(t))+\cos(\text{BR}_F(t))}\right)^{\gamma},
\end{equation}
and for $\text{SR}_F(t)<\theta<\beta_2$,
\begin{equation}\label{Ch3exfsrbr2}
V(\theta,t)=V(\text{SR}_F(t),t)
\left(\frac{(1-\mu)\sin(\theta)+\cos(\theta)}{(1-\mu)\sin(\text{SR}_F(t))+\cos(\text{SR}_F(t))}\right)^{\gamma}.
\end{equation}

\begin{figure}[h]
\centering
\includegraphics[width=12cm,height=5cm]{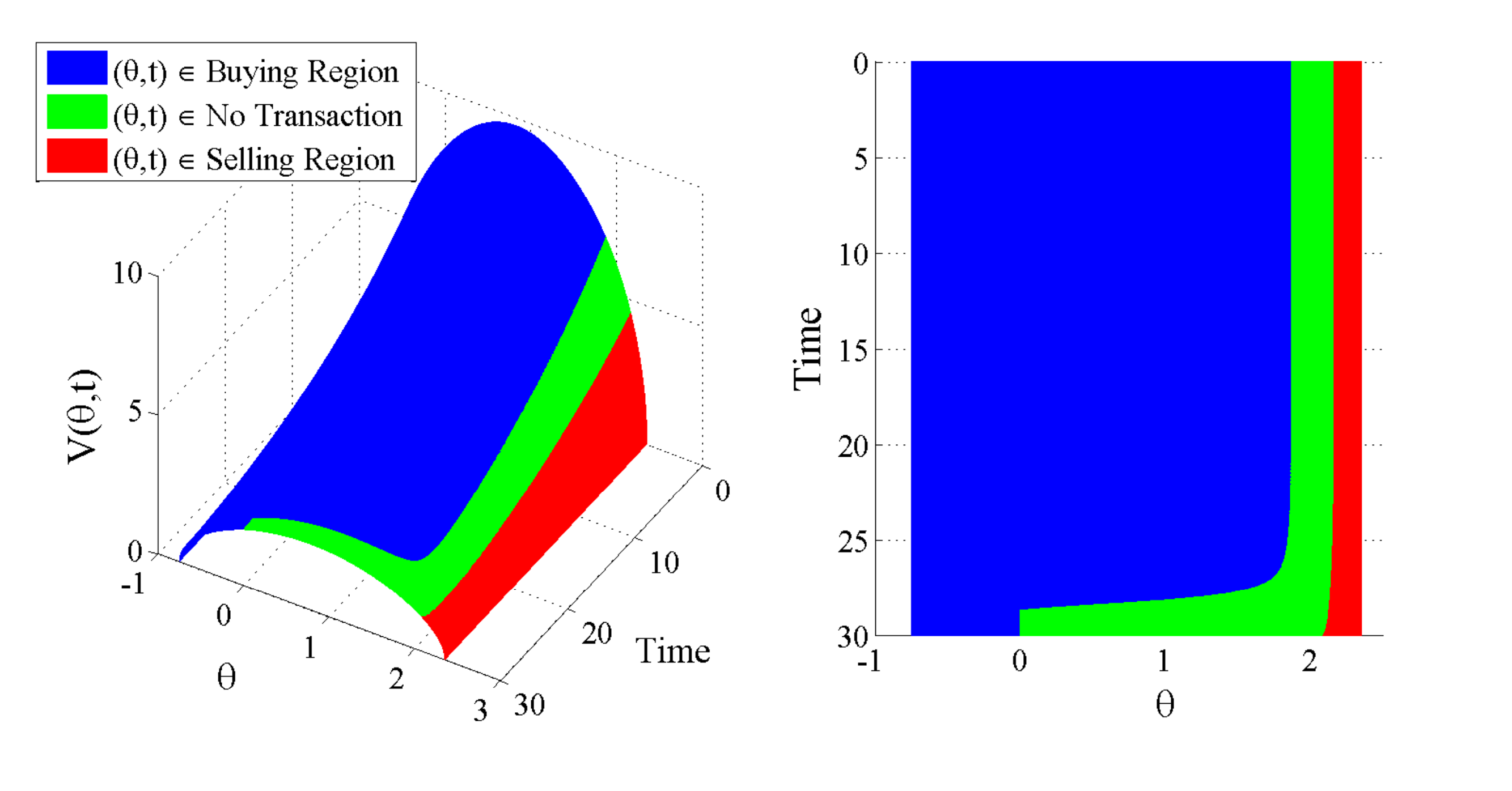}
\caption{\label{Ch3stationaryCN2} Value function (numerical solution) for $(\theta,t)\in[\beta_1,\beta_2]\times[0,30]$. The colour code is blue if $(\theta,t)$ is in the buying region, green in the no transactions region  and red in the selling region.}
\end{figure}

Figure \ref{Ch3stationaryCN2} represents the value function in perspective (left) and from above (right) in $[\beta_1, \beta_2]$,  for a maturity of $T=30$ years. The Figure shows the numerical values obtained for the function $V(\theta,t)$ with the method described in Section \ref{Ch3NM}. We have coloured the function depending in whether $(\theta,t)$ is in the Buying, Selling or No Transactions region. We can visually check the expected monotonicity of the Buying and Selling frontiers studying from above (right) the two curves which divide the different colours (red-green and green-blue). The Buying frontier remains constant ($\text{BR}_F(t)=0$) for a certain period near maturity and the stationarity value of both frontiers as we move away from maturity is also observable.

\section{Numerical Method}\label{Ch3NM}

The numerical method described in this section is constructed upon the following strategy.

Let $\pi\in A_{0}(x,y)$ denote an admissible trading strategy where $x$ and $y$ are the amount of money in the bank and stock accounts at $t=0$.

Let $\alpha_1\in(\beta_1,\text{BR}_s)$ and $\alpha_2\in(\text{SR}_s,\beta_2)$ and
define
\begin{equation}
A^{\alpha_1, \alpha_2}_0(x,y)=\left\{\pi\in A_0(x,y) \ \mid \ \text{arccot}\left({x^\pi}/{y^\pi}\right)\in (\alpha_1, \alpha_2)\right\}
\end{equation}
where  $x^\pi$, $y^{\pi}$ are the amounts in the bank and stock accounts if strategy $\pi$ is followed.

 Proposition~\ref{Ch3tradpropertapolar}  implies that $\pi^{o}\in A^{\alpha_1, \alpha_2}_0(x,y)$ where $\pi^{o}$ denotes the optimal trading strategy solving (\ref{Ch3defvarpphi}).
Therefore, the optimal value function can be computed as the solution of (\ref{Ch3enunprobpolar})-(\ref{Ch3enunprobpolar2}) in $(\alpha_1,\alpha_2)\times[0,T]$ subject to the boundary conditions:
\begin{equation}\label{Ch3theorboundcond}
\begin{aligned}
V_{\theta}\left(\alpha_1,t\right)&=V\left(\alpha_1,t\right)\gamma \frac{(1+\lambda)\cos(\alpha_1)-\sin(\alpha_1)}{(1+\lambda)\sin(\alpha_1)+\cos(\alpha_1)}, \\
V_{\theta}\left(\alpha_2,t\right)&=V\left(\alpha_2,t\right)\gamma\frac{(1-\mu)\cos(\alpha_2)-\sin(\alpha_2)}{(1-\mu)\sin(\alpha_2)+\cos(\alpha_2)}.
\end{aligned}
\end{equation}
These conditions are equivalent to a mandatory buying or selling the stock if $\theta$ reaches $\alpha_1$ or $\alpha_2$ respectively (see formulas (\ref{Chp3BRexpformula})-(\ref{Chp3SRexpformula})).

The solution can be extended to $(\beta_1,\beta_2)\times[0,T]$ taking into account that for $t\in[0,T]$:
$(\beta_1,\alpha_1)\subset \text{BR}$ and
$(\alpha_2,\beta_2)\subset \text{SR}$,
so that we can compute $V(\theta,t)$ with (\ref{Ch3exfsrbr1}) in $\text{BR}$ and (\ref{Ch3exfsrbr2}) in  $\text{SR}$.

\subsection{The adaptive mesh}\label{Ch3MACCMAM}
Let $N_t$ be a nonnegative integer and let us define the time mesh $\left\{t_l\right\}_{l=0}^{N_t}$ by
\begin{equation}\label{Ch3timediscret}
t_l=l\Delta t, \quad l=0,1,...,N_t, \quad \Delta t = \frac{T}{N_t}.
\end{equation}
The spatial mesh will depend on the time step. The main idea is to adapt the mesh in such a way that it evolves through time following the approximate location of the buying and selling frontiers (i.e. evolving as the green zone in Figure \ref{Ch3stationaryCN2}). To this end, let $\delta\in(0, 1/2)$ be a control parameter. We define
\begin{equation}\label{Ch3paramk}
\begin{aligned}
k_1 &=\frac{\beta_2-\text{SR}_s}{\beta_2-\text{BR}_F(T)}, \\
K & =\min\{\delta, k_1, \ \text{BR}_F(T)-\beta_1\},
\end{aligned}
\end{equation}
where $\text{SR}_s=\text{arccot}\left(\text{SR}^c_s\right)\in[\beta_1,\beta_2]$ is the stationary state of the Selling frontier (see Proposition~\ref{Ch3tradpropertapolar}).

For $N\in \mathbb{N}$, let us consider the $N+1$ Chebyshev nodes in $[-1, 1]$,
\begin{equation}\label{Ch3chebynodes}
\tilde{\theta}_j=\cos\left(\frac{\pi j}{N}\right), \quad j=0,1,...,N.
\end{equation}

We define the integer $j_K\in\{0,1,2,...,N\}$ as the unique integer such that
\begin{equation}\label{Ch3consjk}
\left|\tilde{\theta}_{N-j_K}-\tilde{\theta}_{N}\right|\leq 2K < \left|\tilde{\theta}_{N-(j_{K}+1)}-\tilde{\theta}_{N}\right|.
\end{equation}
Note that $j_K$  is well defined because $0<K\leq\delta<1/2$. From the definition of the Chebyshev nodes, it is easy to check that it exists $N_0$ such that for all $N\geq N_0$,  $j_K\geq 1$.

Let us suppose that at time $t=t_l$, we know the locations of the buying and selling frontiers, $BR_F(t_l)$ and $SR_F(t_l)$. Given $N_{\theta}\in\mathbb{N}$, $N_{\theta}>N_0$ we define an interval $I(t_l)$ by:
\begin{equation}\label{Ch3forint1}
I(t_l)=\left[0, \ \frac{2}{\tilde{\theta}_{j_K}-\tilde{\theta}_{N_{\theta}}}\text{SR}_F(t_l) \right],
\end{equation}
if  $\text{BR}_F(t_l)=0$, or, in case $\text{BR}_F(t_l)>0$,
\begin{equation}\label{Ch3forint2}
I(t_l)=\left[\text{BR}_F(t_l)-M\frac{\tilde{\theta}_{N_{\theta}-j_K}-\tilde{\theta}_{N_{\theta}}}{\tilde{\theta}_{j_K}-\tilde{\theta}_{N_{\theta}-j_K}}, \ SR_F(t_l)+M\frac{\tilde{\theta}_0-\tilde{\theta}_{j_K}}{\tilde{\theta}_{j_K}-\tilde{\theta}_{N_{\theta}-j_K}} \right].
\end{equation}
Here $M=\text{SR}_F(t_l)-\text{BR}_F(t_l)$. We remark that, with this definition, the interval $I(t_l)$ always contains the no transaction region
$[\text{BR}_F(t_l),\text{SR}_F(t_l)]$ and it is contained in the solvency region $[\beta_1,\beta_2]$. Furthermore, $\text{BR}_F(t_l)$ and $\text{SR}_F(t_l)$ are always one of the $N_{\theta}+1$ Chebyshev nodes in the interval $I(t_l)$, while restriction $N_{\theta}>N_0$ implies that $\text{SR}_F(t_l)$ is an interior point of $I(t_l)$. The control parameter $\delta$ guaranties that a maximum of $100\delta$\% of interval $I(t_l)$ is contained in the Selling Region, another maximum of $100\delta$\% of $I(t_l)$ in the Buying Region, whereas a minimum $100(1-2\delta)$\% of $I(t_l)$ is contained in the No Transactions Region.

Note also that for $t_{N_t}=T$, the values $\text{BR}_F(T)=0$ and $\text{SR}_F(T)$ are, of course, known data (see Proposition \ref{Ch3tradpropertapolar}), whereas for $t_l<T$, we have to substitute $BR_F(t_l)$ and $SR_F(t_l)$ by some approximation that we will denote $BR^{\textbf{N}}_F(t_l)$ and $SR^{\textbf{N}}_F(t_l)$ where $\textbf{N}=(N_{\theta},N_t)$. We will describe in Subsection \ref{Ch3MACCMCC} how to compute them prior to the construction of the interval $I(t_l)$.

Next proposition proves that, for $N_t$  big enough, $\text{BR}_F(t_{l-1}), \text{SR}_F(t_{l-1})\in I(t_l)$, so that we can compute recursively the intervals $I(t_j)$ for $j=N_t,N_t-1,\dots, 0$.

\begin{proposition}\label{Ch3propoinclusnteninter}

Let $NT(t)=[BR_F(t), \ SR_F(t)]$ where $BR_F(t)$ and $SR_F(t)$ are the exact location of the Buying and Selling frontiers.

For any $N_{\theta}>N_0$, where $N_0$ is the restriction which guarantees that $\text{SR}_F(t)$ will be in the interior of $I(t)$, compute $I(t)$ with (\ref{Ch3forint1}) or (\ref{Ch3forint2}).

It exists $N_1>0$ such that for any time mesh $\left\{t_l\right\}_{l=0}^{N_t}, \ N_t>N_1$ given by (\ref{Ch3timediscret}), it holds
\begin{equation}
NT(t_{l_0-1}), \ NT(t_{l_0})\subset I(t_{l_0}).
\end{equation}
for any $t_{l_0}\in\left\{t_l\right\}_{l=0}^{N_t}$.
\end{proposition}

\begin{proof}

From \cite{DaiYi}, we know that $\text{SR}^c_F(t)\in \mathscr{C}^{\infty}[0,T)$. Therefore, $\text{SR}_F(t)=\text{arccot}\left(\text{SR}^c_F(t)\right)\in(\beta_1, \ \beta_2)$ is $\mathscr{C}^{\infty}[0,T)$.

Let $k$ from (\ref{Ch3paramk}) be fixed. Since $\text{SR}_F(t)$ is in the interior of $I(t)$, it will exist $\Delta t_{k}$ such that for all $ \Delta t<\Delta t_{k}$:
\begin{equation}\label{Ch3condicionencontrarfron}
SR_F(t-\Delta t)\in I(t), \quad t \in [0,  T).
\end{equation}

This guarantees that for any equally spaced time mesh  $\{t_l\}_{l=0}^{N_t}$, with $N_t>1/\Delta t_{k}$, $SR_F(t_{l-1})\in I(t_l)$.

To finish the proof, note that from Proposition \ref{Ch3tradpropertapolar}, $\text{BR}_F(t_{l})\leq\text{BR}_F(t_{l-1})$ and that $\text{BR}_F(t_{l-1})\leq \text{SR}_F(t_{l-1})$, so the result follows directly from the definition of $I(t_l)$.
\end{proof}

\subsection{Chebyshev collocation Method.}\label{Ch3MACCMCC}
Let us suppose that we know an approximation of the function value $V^{\textbf{N}}(\theta,t_{l})$, $\theta\in (\beta_1, \beta_2)$ and approximate values of $\text{BR}^{\textbf{N}}_F(t_l)$ and $\text{SR}^{\textbf{N}}_F(t_l)$ at time $t=t_l$. For $N_{\theta}$ big enough (Proposition~\ref{Ch3propoinclusnteninter}), we can compute $I(t_l)=[\alpha^{t_l}_1, \alpha^{t_l}_2]$ defined as in (\ref{Ch3forint1}) if $\text{BR}^{\textbf{N}}_F(t_l)=0$, or with (\ref{Ch3forint2}) otherwise.

For $t\in[t_{l-1},t_{l}]$, we define the function $\hat{V}$ as the function value which gives the expected terminal value when the trading strategy is to perform no transactions if $\theta\in(\alpha^{t_l}_1, \alpha^{t_l}_2)$, to buy the stock if $\theta=\alpha^{t_l}_1$ and to sell the stock if $\theta=\alpha^{t_l}_2$, subject to $\hat{V}(\theta,t_{l})=V^{\textbf{N}}(\theta,t_{l})$.

Therefore, $\hat{V}$ is the solution of the equation
\begin{equation}\label{Ch3ecudifnotranbis}
-\hat{V}_t+g_2(\theta)\hat{V}_{\theta \theta}+g_1(\theta)\hat{V}_{\theta}+g_0(\theta) \hat{V}=0,
\end{equation}
subject to
\begin{equation}\label{Ch3ecudifnotranboundcondbis}
\begin{aligned}
\hat{V}_{\theta}\left(\alpha^{t_l}_1,t\right)&=\hat{V}\left(\alpha^{t_l}_1,t\right)\gamma \frac{(1+\lambda)\cos(\alpha^{t_l}_1)-\sin(\alpha^{t_l}_1)}{(1+\lambda)\sin(\alpha^{t_l}_1)+\cos(\alpha^{t_l}_1)}, \\
\hat{V}_{\theta}\left(\alpha^{t_l}_2,t\right)&
=\hat{V}\left(\alpha^{t_l}_2,t\right)\gamma\frac{(1-\mu)\cos(\alpha^{t_l}_2)
-\sin(\alpha^{t_l}_2)}{(1-\mu)\sin(\alpha^{t_l}_2)+\cos(\alpha^{t_l}_2)}, \\
\hat{V}(\theta,t_l)&=V^{\textbf{N}}(\theta,t_{l}).
\end{aligned}
\end{equation}

Let us consider the $N_{\theta}+1$ Chebyshev nodes in $I(t_l)$
\begin{equation}\label{Ch3cambvariable}
\theta_j=\frac{\alpha^{t_l}_2-\alpha^{t_l}_1}{2}\tilde{\theta}_j+\frac{\alpha^{t_l}_2+\alpha^{t_l}_1}{2}, \ j=0,1,...,N_{\theta},
\end{equation}
where $\tilde{\theta}_j$ are the Chebyshev points (\ref{Ch3chebynodes}).

The numerical approximation $\hat{V}^{\textbf{N}}(\theta,t_{l-1})$, $\theta\in(\theta_{N_{\theta}},\theta_0)$ to the function $\hat{V}$ is the collocation polynomial \cite{Canuto} of degree $N_{\theta}$ defined for $j=1,...,N_{\theta}-1$ by:
\begin{equation}\label{Ch3collocationmethod}
\frac{\hat{V}^{\textbf{N}}(\theta_j,t_{l-1})-\hat{V}^{\textbf{N}}(\theta_j,t_l)}{\Delta t}
=L\left(\frac{\hat{V}^{\textbf{N}}(\theta_j,t_{l-1})+\hat{V}^{\textbf{N}}(\theta_j,t_l)}{2}\right).
\end{equation}
subject to
\begin{equation}\label{Ch3valorcondvencheby}
\hat{V}^{\textbf{N}}(\theta_j,t_l)=V^{\textbf{N}}(\theta_j,t_l), \quad j=1,2,...,N_{\theta}-1,
\end{equation}
with (Neumann) boundary conditions
\begin{equation}\label{Ch3collocationmethodcons}
\begin{aligned}
\hat{V}^{\textbf{N}}_{\theta}(\theta_{N_{\theta}},t_{l-1}) &=V^{\textbf{N}}(\theta_{N_{\theta}},t_{l})\gamma \frac{(1+\lambda)\cos(\theta_{N_{\theta}})
-\sin(\theta_{N_{\theta}})}{(1+\lambda)\sin(\theta_{N_{\theta}})+\cos(\theta_{N_{\theta}})}, \\
\hat{V}^{\textbf{N}}_{\theta}(\theta_0,t_{l-1}) &=V^{\textbf{N}}(\theta_0,t_l)\gamma\frac{(1-\mu)\cos(\theta_0)
-\sin(\theta_0)}{(1-\mu)\sin(\theta_0)+\cos(\theta_0)}.
\end{aligned}
\end{equation}
where
\begin{equation}\label{Ch3operatorl}
L\bigl(\hat{V}^{\textbf{N}}(\theta)\bigr)=
g_2(\theta)\frac{\partial^2 \hat{V}^{\textbf{N}}}{\partial \theta^2}+ g_1(\theta)\frac{\partial \hat{V}^{\textbf{N}}}{\partial \theta}+g_0(\theta) \hat{V}^{\textbf{N}}.
\end{equation}

The equations (\ref{Ch3collocationmethod})-(\ref{Ch3collocationmethodcons}) define a dense system of linear equations to find the values of $\hat{V}^{\textbf{N}}_{\theta_j}$, $j=0,\dots, N_\theta$.
 However, the fact that, with relative few nodes for the spatial mesh we can achieve a very good precision, makes this method competitive with respect to a finite differences method, see Section~\ref{Ch3numresults}.

 Let us define
\begin{equation}\label{Ch3auxencheby}
\begin{aligned}
P^{(\textbf{N},l-1)}_1(\theta) &= \hat{V}^{\textbf{N}}_{\theta}(\theta,t_{l-1}) - \hat{V}^{\textbf{N}}(\theta,t_{l-1})\cdotp \gamma\frac{(1+\lambda)\cos(\theta)-\sin(\theta)}{(1+\lambda)\sin(\theta)+\cos(\theta)}, \\
P^{(\textbf{N},l-1)}_2(\theta) &=\hat{V}^{\textbf{N}}_{\theta}(\theta,t_{l-1}) -\hat{V}^{\textbf{N}}(\theta,t_l)\cdotp \gamma\frac{(1-\mu)\cos(\theta)-\sin(\theta)}{(1-\mu)\sin(\theta)+\cos(\theta)},
\end{aligned}
\end{equation}
which are explicit functions because $\hat{V}^{\textbf{N}}$ is a known polynomial in $\theta$. In (\ref{Ch3auxencheby}), we compare (see  \cite[Subsection 3.5.3]{Gaton}) whether it is better to not perform transactions or to buy the stock (resp. sell the stock). If  polynomial $P^{(\textbf{N},l-1)}_1>0$ (resp. $P^{(\textbf{N},l-1)}_2>0$) it is better to not perform transactions rather than buy (resp. sell) the stock.

The numerical approximation to the Buying and Selling frontiers is given by:
\begin{equation}\label{frontiers}
\begin{aligned}
\text{BR}^{\textbf{N}}_F(t_{l-1}) &= \min\left\{\beta: P^{(\textbf{N},l-1)}_1(\theta)\geq 0, \theta\in[\beta,\alpha^{t_l}_2) \right\}, \\
\text{SR}^{\textbf{N}}_F(t_{l-1}) &= \max\left\{\beta: P^{(\textbf{N},l-1)}_2(\theta)\geq 0, \theta\in(\alpha^{t_l}_1,\beta] \right\}, \\
\end{aligned}
\end{equation}

Once we know the location of the frontiers and the function value in that points, we can compute the approximate function value through the following explicit formulas where we have used the notation $B_{l}=\text{BR}^{\textbf{N}}_F(t_{l})$, $S_{l}=\text{SR}^{\textbf{N}}_F(t_{l})$
\begin{equation}\label{VN1}
V^{\textbf{N}}(\theta^{t_{l-1}}_j,t_{l-1})
=\hat{V}^{\textbf{N}}(B_{l-1},t_{l-1})
\left[\frac{(1+\lambda)\sin(\theta^{t_{l-1}}_j)
+\cos(\theta^{t_{l-1}}_j)}{(1+\lambda)\sin(B_{l-1})+\cos(B_{l-1})}\right]^{\gamma},
\end{equation}
if $ \theta^{t_{l-1}}_j<B_{t_{l-1}}$,
\begin{equation}\label{VN2}
V^{\textbf{N}}(\theta^{t_{l-1}}_j,t_{l-1})=\hat{V}^{\textbf{N}}(\theta^{t_{l-1}}_j,t_{l-1}), \quad
B_{l-1}\leq\theta^{t_{l-1}}_j\leq S_{l-1},
\end{equation}
\begin{equation}\label{VN3}
V^{\textbf{N}}(\theta^{t_{l-1}}_j,t_{l-1})=
\hat{V}^{\textbf{N}}(S_{t_{l-1}},t_{l-1})
\left[\frac{(1-\mu)\sin(\theta^{t_{l-1}}_j)
+\cos(\theta^{t_{l-1}}_j)}{(1-\mu)\sin(S_{l-1})+\cos(S_{l-1})}\right]^{\gamma},
\end{equation}
if $\theta^{t_{l-1}}_j>S_{l-1}$.

Then the complete algorithm reads as follows:
\begin{itemize}
\item[\textbf{Step 0}] Fix a number $N_t$ and a number $N_{\theta}$ big enough such that Proposition~\ref{Ch3propoinclusnteninter} holds.

Compute $\Delta t = \frac{T}{N_t}$ and $\{t_l\}_{l=0}^{N_t}$ as in (\ref{Ch3timediscret}). Define $\textbf{N}=(N_{\theta},N_t)$.

Set $l=N_{t}$ and  compute $I(t_{N_t})$ with formula (\ref{Ch3forint1})

Compute $V^{\textbf{N}}(\theta,T)$, $\theta\in I(t_{N_t})$, as the Chebyshev interpolation polynomial in $\{\theta^{T}_j\}_{j=0}^{N_{\theta}}$ of function $V(\theta,T)$, given by (\ref{Ch3enunprobpolar2}), where
$\{\theta^{T}_j\}_{j=0}^{N_{\theta}}$ denote the Chebyshev nodes in $I(t_{N_t})$.
\item[\textbf{Step 1}] Compute the polynomial $\hat{V}^{\textbf{N}}(\theta,t_{l-1})$ solving the collocation equations (\ref{Ch3collocationmethod}) with final condition (\ref{Ch3valorcondvencheby}) and boundary conditions (\ref{Ch3collocationmethodcons}).
\item[\textbf{Step 2}] Locate the buying and selling frontiers $\text{BR}^{\textbf{N}}_F(t_{l-1})$ and
$\text{SR}^{\textbf{N}}_F(t_{l-1}) $ using (\ref{frontiers}).
\item[\textbf{Step 3}] Compute the interval $I(t_{l-1})$ with (\ref{Ch3forint1}) if $\text{BR}^{\textbf{N}}_F(t_{l-1})=0$ or with (\ref{Ch3forint2})  otherwise.

     Compute
    the numerical approximation $V^{\textbf{N}}$  at time $t_{l-1}$ with formulae (\ref{VN1}), (\ref{VN2}) and (\ref{VN3}).
    \item[\textbf{Step 4}] Set $l=l-1$ and stop if $l=0$ or, otherwise, proceed to \textbf{Step 1}.
\end{itemize}

\begin{remark}\normalfont{
In the algorithm we propose there is an error related to the imposition of Neumann boundary conditions in (\ref{Ch3collocationmethodcons}) instead of the Robin type correct ones. This error can be controlled by the size of the discretization parameters $N_t$ and $N_{\theta}$ because, by definition of the adaptive interval, $\theta_0$ is always inside the Selling Region and $\theta_{N_{\theta}}$ is inside the Buying Region or it is the Buying Frontier. We point that for $t\in[\hat{t}_0, \ T]$,  the lower limit of the interval $I(t)$ is $\alpha_1(t)=0$, which is the Buying Frontier, so that,  it is not inside the Buying Region. Nevertheless, note that when we compute function $\hat{V}$, the boundary condition at $\alpha_1(t)=0, \ t\in[\hat{t}_0,T]$ must be mandatorily to buy the stock, so that
\begin{equation*}
V_{\theta}(0,t)=\underset{\theta\rightarrow 0^-}\lim V_{\theta}(\theta,t)=V(0,t)\gamma \frac{(1+\lambda)\cos(0)-\sin(0)}{(1+\lambda)\sin(0)+\cos(0)}
\end{equation*}
}
\end{remark}

\section{Numerical Results}\label{Ch3numresults}

We consider the parameter values as in the first experiment in \cite{Arregui}. For $t\in[0, 4]$ let:
\begin{equation*}
\sigma=0\ldotp25, \ \ r=0\ldotp03, \ \ \alpha=0\ldotp10, \ \ \gamma=0\ldotp5, \ \ \lambda=0\ldotp08, \ \ \mu=0\ldotp02,
\end{equation*}

The following figure shows the numerical solution $V^{\textbf{N}}(\theta,t)$.

\begin{figure}[h]\label{Ch3v1Valfun4anos}
\centering
\includegraphics[width=12cm,height=5cm]{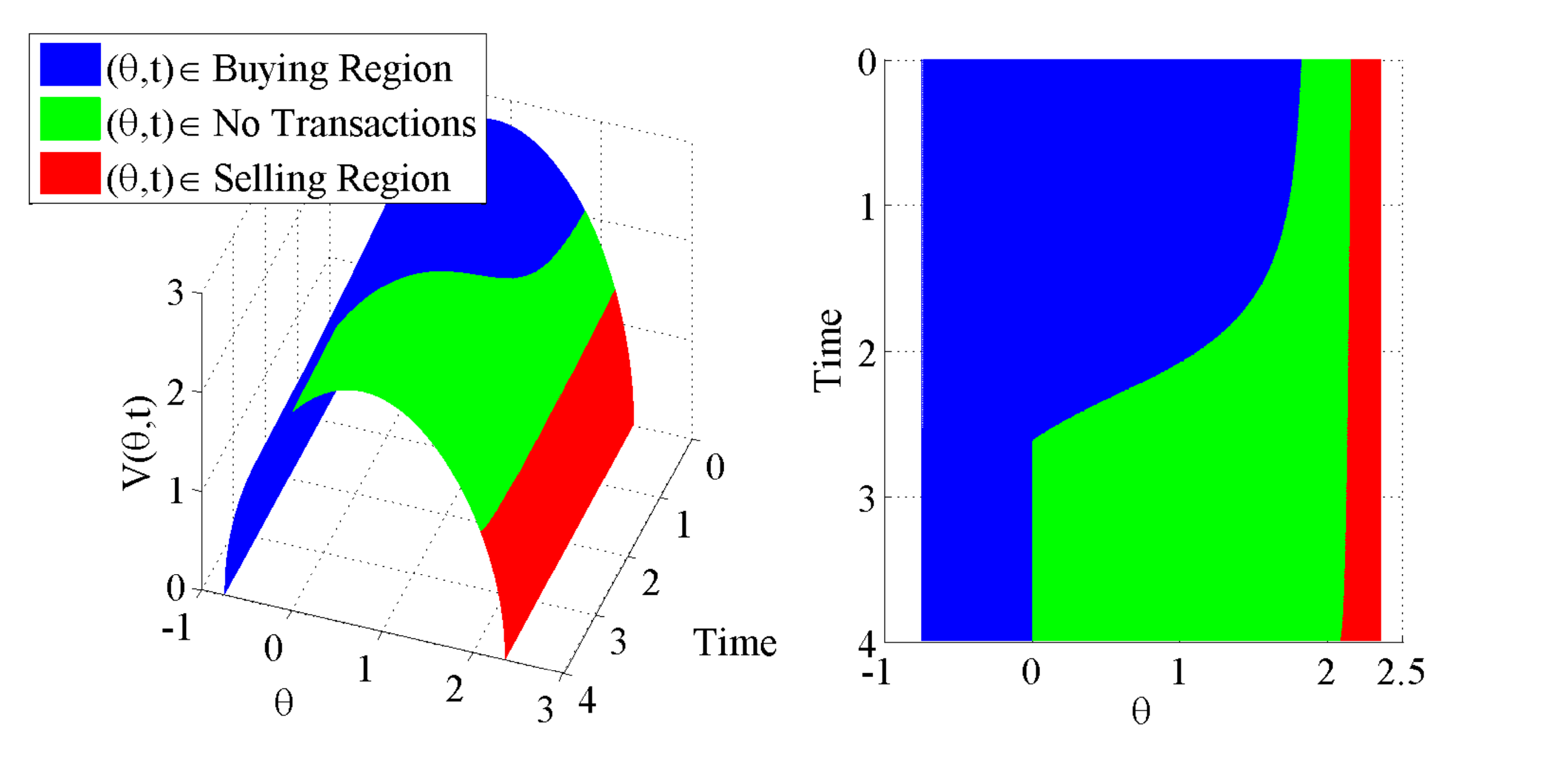}
\caption[Value function in polar coordinates for $ t \in {[0, 4]} $ ]{Value of $V^{\textbf{N}}(\theta,t)$ for $t\in{[0, 4]}$. The colour code is blue if $(\theta,t)$ is in the buying, green in the no transactions and red in the selling region.}
\end{figure}

We have colored the function depending in whether $(\theta,t)$ is in the Buying, Selling or No Transactions region. As in Figure \ref{Ch3stationaryCN2}, we can visually check the properties from Proposition \ref{Ch3tradpropertapolar}.

First, we establish the criteria employed in the experiments to build the spatial mesh. We have fixed the control parameter $\delta=0\ldotp1$, so that, at least $80$\% of the interval corresponds to the No Transactions Region. The particular choice of $\delta$ does not affect the rate of convergence of the error.

In order to compare the performance of the spectral method with other numerical methods, we have also implemented a Central Differences (CD) based method in order to solve the PDE in Step 2 (see Subsection \ref{Ch3MACCMCC}). The formal study of the error will be conducted for the cases where explicit formulas are available, comparing the results of the Central differences and Chebyshev methods. The rest of the properties given in \cite{DaiYi}, although not included, were also checked.

\subsection{Value of the function in $v(0,t)$}\label{Ch3numresultsvalpimed}

We consider $v(z,t)$ defined in (\ref{Checamborig}). For $z=0$,  we can explicitly compute $v(0,t)$ with \cite[(3.9)]{DaiYi}. In Figure \ref{Ch3analiticalsolv0t} we plot the value of $v(0,t)$ for $t\in[0,4]$.

\begin{figure}[h]
\centering
\includegraphics[width=6.5cm,height=5cm]{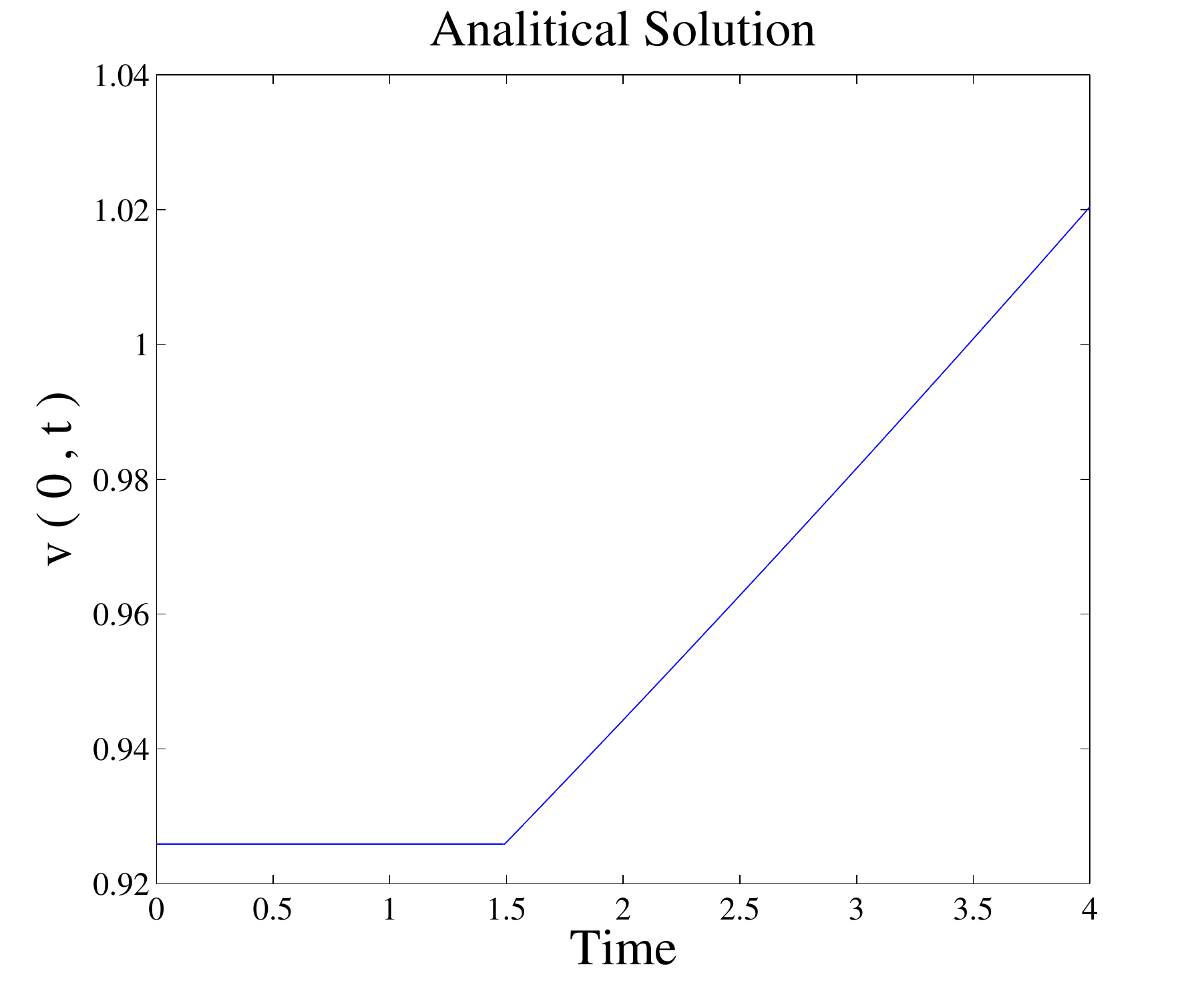}
\caption[Analytical solution of $v(0,t)$]{\label{Ch3analiticalsolv0t} Analytical solution of $v(0,t), \ t\in[0,4]$.}
\end{figure}

The value $z=0$ corresponds in polar coordinates to $\theta=\frac{\pi}{2}$. A numerical solution $v^{\textbf{N}}(0,t_i)$ can be computed explicitly using $V^{\textbf{N}}\left(\frac{\pi}{2},t_i\right)$ and formula (\ref{Ch3formularelacionpolarorig}), which relates the function in polar coordinates and in the original variables.
\begin{equation*}
\begin{aligned}
v^{\textbf{N}}(z,t) &=-\left(\frac{V^{\textbf{N}}_{\theta}(\theta,t)\sin^2(\theta)-\gamma\sin(\theta)\cos(\theta)V^{\textbf{N}}(\theta,t)}{\gamma V^{\textbf{N}}(\theta,t)}\right),  \\
z &=\cot(\theta).
\end{aligned}
\end{equation*}

The following Figure compares the difference between the analytical solution $v(0,t), \ t\in[0,4]$ and the numerical solution obtained with the Chebyshev method for $N_{\theta}=256$ (left) and $N_{\theta}=2048$ (right). Both pictures are in the same scale and we can observe that the error reduces for increasing value of $N_{\theta}$.

\begin{figure}[h]
\centering
\includegraphics[width=14cm,height=5.5cm]{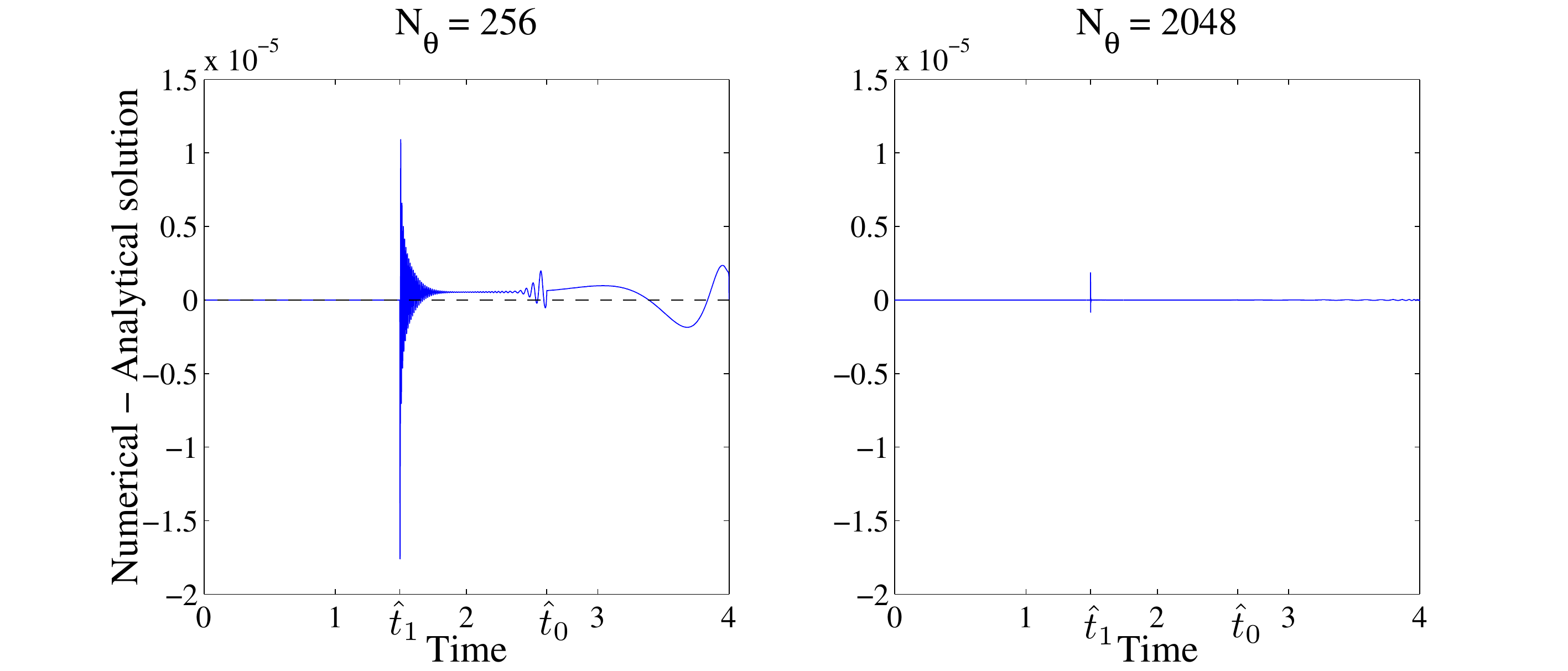}
\caption[Difference between the numerical and analytical solution III]{\label{Ch3valpimedcomcheb} Value $v(0,t)-v^{\textbf{N}}(0,t), \ t\in[0,4]$ where $v^{\textbf{N}}$ was computed with the Chebyshev method with $N_{\theta}=256$ (left) and $N_{\theta}=2048$ (right).}
\end{figure}

In both pictures of Figure \ref{Ch3valpimedcomcheb} we can observe an error discontinuity at time $\hat{t}_1$. From \cite[(3.9)]{DaiYi}, we know that the function $v(0,t)$ is not derivable (respect time) at instant $\hat{t}_1$. The same phenomena can be observed in the numerical experiments in \cite{Arregui}. We can also see that some oscillations appear at time $\hat{t}_0$ where we change the kind of adaptive mesh $I(t_i)$ (see Subsection \ref{Ch3MACCMAM}).

We proceed to check the rate of error convergence. We define the Root of the Mean Square Error as
\begin{equation}\label{Ch3errrorcuadmed}
\text{RMSE}_{\{N_{\theta},N_t\}}\left(v^{\textbf{N}}\right)=\sqrt{\frac{1}{N_t+1}\sum_{l=0}^{N_t} \left(v^{\textbf{N}}(0,t_l)-{v}(0,t_l)\right)^2}.
\end{equation}

Figure \ref{Ch3converrtesptempvalpimed} shows the convergence of spatial error (left) for $\Delta t=3\ldotp9 \cdotp 10^{-4}$ and different number of spatial nodes $N_{\theta}$. The right side shows the convergence of temporal error for $N_{\theta}$ fixed and different values of $N_t$.

\begin{figure}[h]
\centering
\includegraphics[width=12.5cm,height=5.3cm]{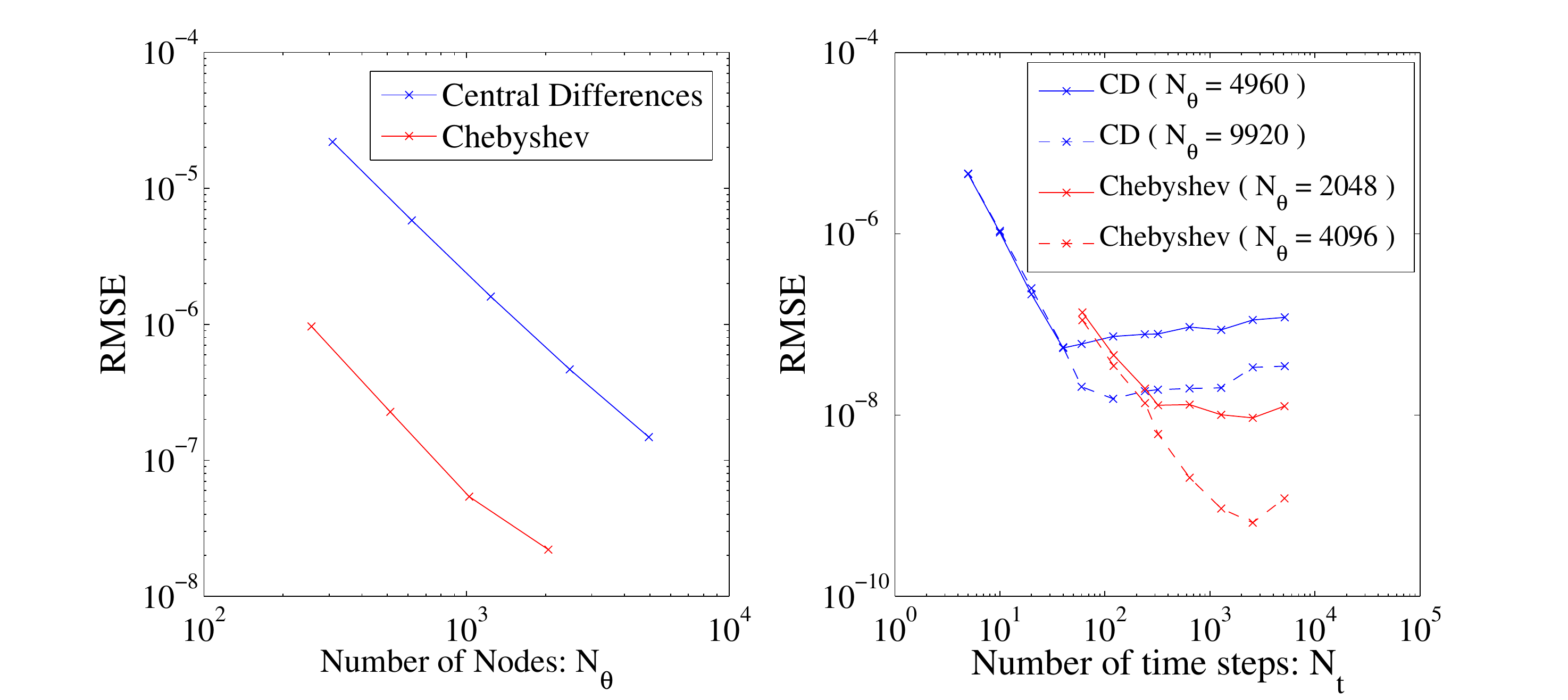}
\caption[Spatial and temporal error convergence valpimed]{\label{Ch3converrtesptempvalpimed} Spatial (left) and Temporal (right) Error convergence of $v^{\textbf{N}}$ in logarithmic scale of the Central Differences (blue) and Chebyshev (red) methods.}
\end{figure}

In the left side of Figure \ref{Ch3converrtesptempvalpimed} we have plotted, in logarithmic scale, the number $N_{\theta}$ of spatial nodes versus the value $\text{RMSE}_{\{N_{\theta},N_t\}}\left(v^{\textbf{N}}\right)$. The slope of the regression line of the CD method (plotted in blue) is $-1\ldotp 80$ and of the Chebyshev method (plotted in red)  is $-1\ldotp85$. The spectral convergence that we could expect in the Chebyshev method  does not occur due to the regularity of the problem.

In the right hand side of Figure \ref{Ch3converrtesptempvalpimed}, we have plotted, in logarithmic scale, the number $N_{t}$ of time steps versus the value $\text{RMSE}_{\{N_{\theta},N_t\}}\left(v^{\textbf{N}}\right)$. The slope of the regression line of the CD method (solid-blue) is $-2\ldotp 31$ as it could be expected from an order 2 method. The slope of the Chebyshev method (solid-red)  is $-1\ldotp4$. We note that for large values of $N_t$ we reach very soon the error limit marked by the size of $N_{\theta}$.

We carry out a second experiment doubling the value of $N_{\theta}$ (right-dashed-blue/red) to check that the lowest value reached by the temporal error was given by the size of the spatial mesh.

Depending on the error tolerance, we might need a big value for $N_{\theta}$ in the CD method but much smaller in the Chebyshev method. This makes that, depending on the required precision, Chebyshev performs better in computational cost than CD. This will be studied below.

\subsection{Location of the Buying Region frontier at time $\hat{t}_1$}\label{Ch3numresultsfronpimed}

From Proposition \ref{Ch3tradpropertapolar}, we know that in polar coordinates $\text{BR}_F(\hat{t}_1)=\frac{\pi}{2}$.

Given a number of time steps $N_t$, we look for $t_{l_1}\in\{t_l\}_{l=0}^{N_t}$ which is nearest to $\hat{t}_1$ and define the Absolute Error (just for this experiment) as:
\begin{equation*}
\text{Absolute Error}_{\textbf{N}}(\hat{t}_1)=\left|\text{BR}^{\textbf{N}}_F(t_{l_1})-\frac{\pi}{2}\right|.
\end{equation*}

The next figure shows the convergence of spatial error (left) for $\Delta t=3\ldotp9 \cdotp 10^{-4}$ and different number of spatial nodes $N_{\theta}$. The right side shows the convergence of temporal error for $N_{\theta}=2048$ (Chebyshev) and $4960$ (CD) and different values of $N_t$.

\begin{figure}[h]
\centering
\includegraphics[width=12.5cm,height=5.3cm]{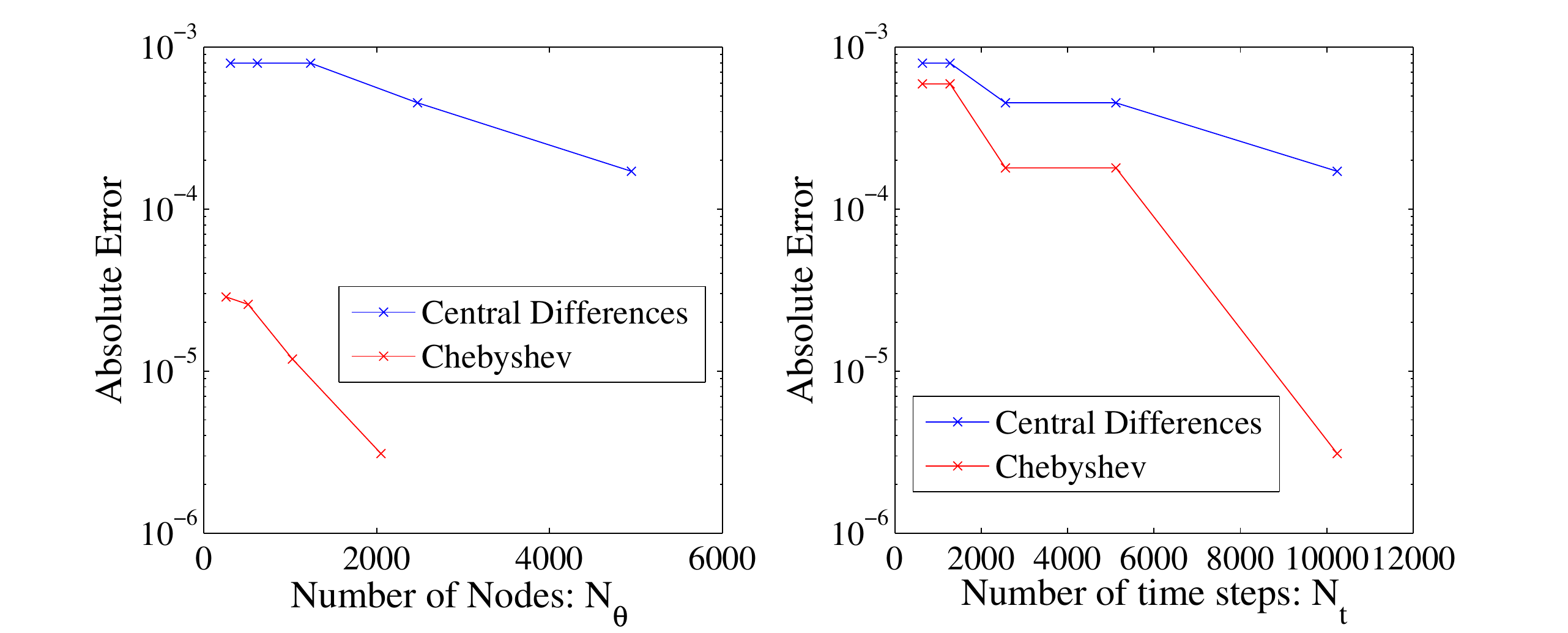}
\caption[Spatial and temporal error convergence fronpimedpimed]{\label{Ch3converrtesptempfronpimed} Spatial (left) and Temporal (right) Error (semilogarithmic scale) of instant when $BR_F=\frac{\pi}{2}$ with the CD (blue) and Chebyshev (red) methods.}
\end{figure}

The spatial error (left) reduces as we increase the value of $N_{\theta}$. At equal number of nodes, the Chebyshev method gives much smaller errors than the CD method.

Concerning the temporal error, the results are step shaped because of the definition of Absolute Error and the time partition when $\Delta t$ is halved. Each time partition is included in the following one and $t_{i_1}$ sometimes changes and sometimes not. The temporal error reduces as we increase the value of $N_{t}$. As in the spatial error, the Chebyshev method outperforms the CD method.

\subsection{First instant when is optimal to have a positive amount of the stock.}\label{Ch3numresultsfronnocero}

From Proposition \ref{Ch3tradpropertapolar}, we know that $\text{BR}_F(t)=0, \quad t\geq \hat{t}_0$ where $\hat{t}_0$ is explicitly computable.

Given a number of time steps $N_t$, we look for $t_{l_0}\in\left\{t_l\right\}_{l=0}^{N_t}$ such that
\begin{equation*}
t_{l_0}\geq \hat{t}_0 > t_{l_0+1}
\end{equation*}

For the Chebyshev method, the $\text{BR}^{\textbf{N}}_F$ may be bigger than 0 a few time steps prior to $l_0$. We note that in the Chebyshev method, the lower limit of $I(t_l), \ t_l\in[\hat{t}_0,T]$ is the Buying frontier.

In left picture of Figure \ref{Ch3v1BuyingfrontierCheby}, we have plotted the numerical estimation of the Buying Frontier with the Chebyshev method for, $N_{\theta}=256$ (blue), $N_{\theta}=512$ (red), $N_{\theta}=1024$ (green) $N_{\theta}=2048$ (black). In the right picture we zoom around $\hat{t}_0$.

\begin{figure}[h]
\centering
\includegraphics[width=13cm,height=5.5cm]{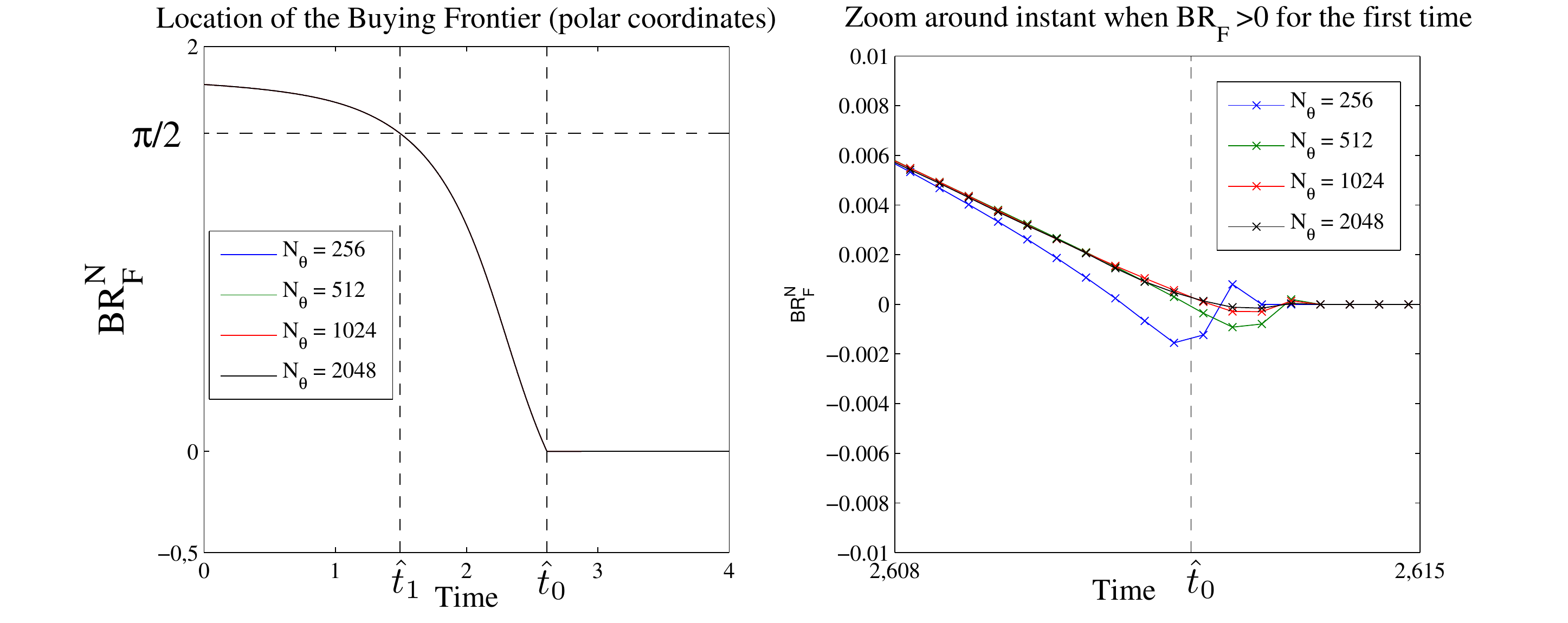}
\caption[Numerically computed Buying Frontier with the Chebyshev method]{\label{Ch3v1BuyingfrontierCheby} Numerically computed Buying Frontier with the Chebyshev method for $t\in[0,4]$ (left) and zoom around $\hat{t}_0$ (right).}
\end{figure}

Let $k\geq 0$ be the biggest value such that
\begin{equation*}
\text{BR}^{\textbf{N}}_F(t_{l_0+k})> 0.
\end{equation*}

If $k>0$, the location of the Buying Frontier oscillates around 0 for $t_l\in\{t_{l_0+k}, ..., t_{l_0+1}\}$ and for $t_{l}< t_{l_0}$ when it behaves as we could expect from Proposition \ref{Ch3tradpropertapolar}.

Numerical experiments show that it is better to let $\text{BR}^{\textbf{N}}_F(t_l)$ oscillate around 0 rather than imposing $\text{BR}^{\textbf{N}}_F(t_l)=\max\{\text{BR}^{\textbf{N}}_F(t_l), \ 0\}$.

The oscillations observed in Figure \ref{Ch3v1BuyingfrontierCheby} are generated by the imposition of the Neumann conditions. The boundary error is controlled by $N_t$ and $N_{\theta}$, but  the spatial error is dominant in this experiment. The instant when the numerical solution begins to oscillate is always very close to $\hat{t}_0$ $\left(\left|t_{l_0-k}-\hat{t}_0\right|\leq 1\ldotp5\cdotp10^{-3}\right) $ and the size of the oscillations reduces as $N_{\theta}$ increases.

These oscillations are the error that we are going to study. They include all the negative values (since the Buying Frontier must be always positive) and any positive value for discrete times larger than $\hat{t}_0$. Thus, we define, for this method and experiment, the absolute error (AE) as
\begin{equation*}
\text{AE}^{Ch}=\max\left\{\left|\underset{l=0,1,...,N_t}{\min}\left\{\text{BR}^{\textbf{N}}_F(t_l)\right\}\right| \ , \left|\underset{l=l_0+1,l=l_0+2,...,N_t}{\max}\left\{\text{BR}^{\textbf{N}}_F(t_l) \right\}\right|\right\}.
\end{equation*}

We fix $\Delta t=3\ldotp9 \cdotp 10^{-4}$ and compute the absolute error for several values for $N_{\theta}$). In Figure  \ref{Ch3v1ConvespfronnoceroCheby} we plot, in logarithmic scale, the value of $N_{\theta}$ versus the absolute error. As we can see the error is rapidly reduced by increasing $N_\theta$.

\begin{figure}[h]
\centering
\includegraphics[width=6cm,height=5cm]{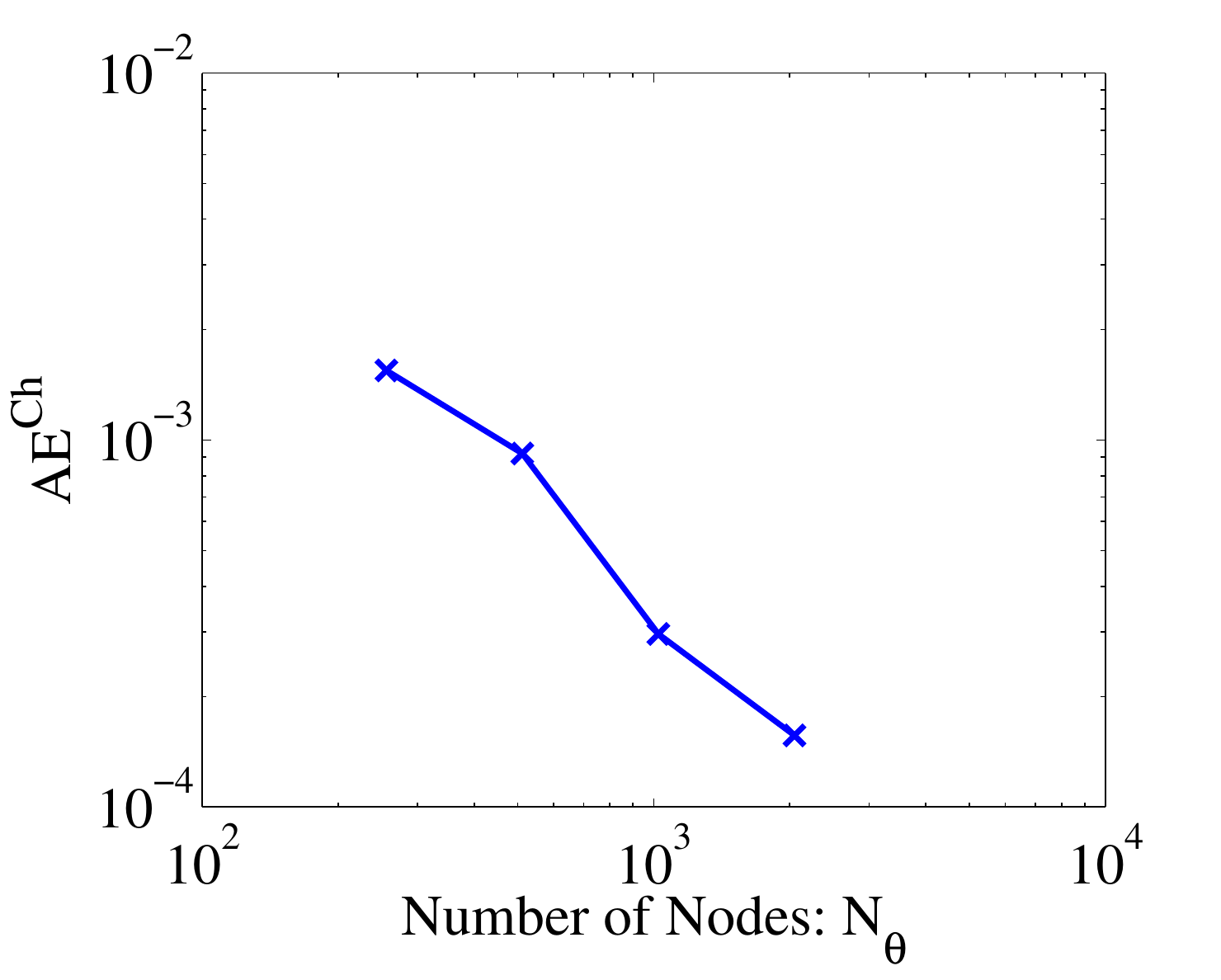}
\caption[Spatial Error convergence (Chebyshev method) when it is optimal to start buying the stock]{\label{Ch3v1ConvespfronnoceroCheby} Spatial error convergence of the first instant when it is optimal to have a positive amount of stock (Chebyshev method).}\label{figura}
\end{figure}

\subsection{Stationary state}

$BR_F$ and $SR_F$ tend to a stationary state as $T\rightarrow \infty$ that can also be computed explicitly (see Proposition \ref{Ch3tradpropertapolar}). Computed with the same model parameters as before but for $T=30$ years (see Figure \ref{Ch3stationaryCN2}), frontiers have stabilized a few years before reaching $t=0$ at:
\begin{equation*}
\begin{aligned}
& \text{Buying Frontier:} \ 1\ldotp8626 \ \ (1\ldotp8622 \ \text{exact value}) \\
& \text{Selling Frontier:} \ 2\ldotp1559 \ \ (2\ldotp1561 \ \text{exact value})
\end{aligned}
\end{equation*}
computed with the Chebyshev method ($\Delta t=10^{-4}$, $N_{\theta}=512$).

We define the absolute error (for this experiment) as
\begin{equation*}
\text{Absolute Error}=\left|\text{BR}^{\textbf{N}}_F(0)-BR_s\right|
\end{equation*}

We study the spatial ($\Delta_t=10^{-3}$ and several values for $N_{\theta}$) and temporal ($N_{\theta}=4960$ for the CD, $N_{\theta}=512$ for the Chebyshev method, and several values for $N_t$) error convergence. In Figure \ref{Ch3v1estacconv} we plot, in logarithmic scale, the value of $N_\theta$ (left) versus the absolute value of the error and the value of $N_t$ (right) versus the absolute value of the error for both methods.

\begin{figure}[h]
\centering
\includegraphics[width=13.5cm,height=6cm]{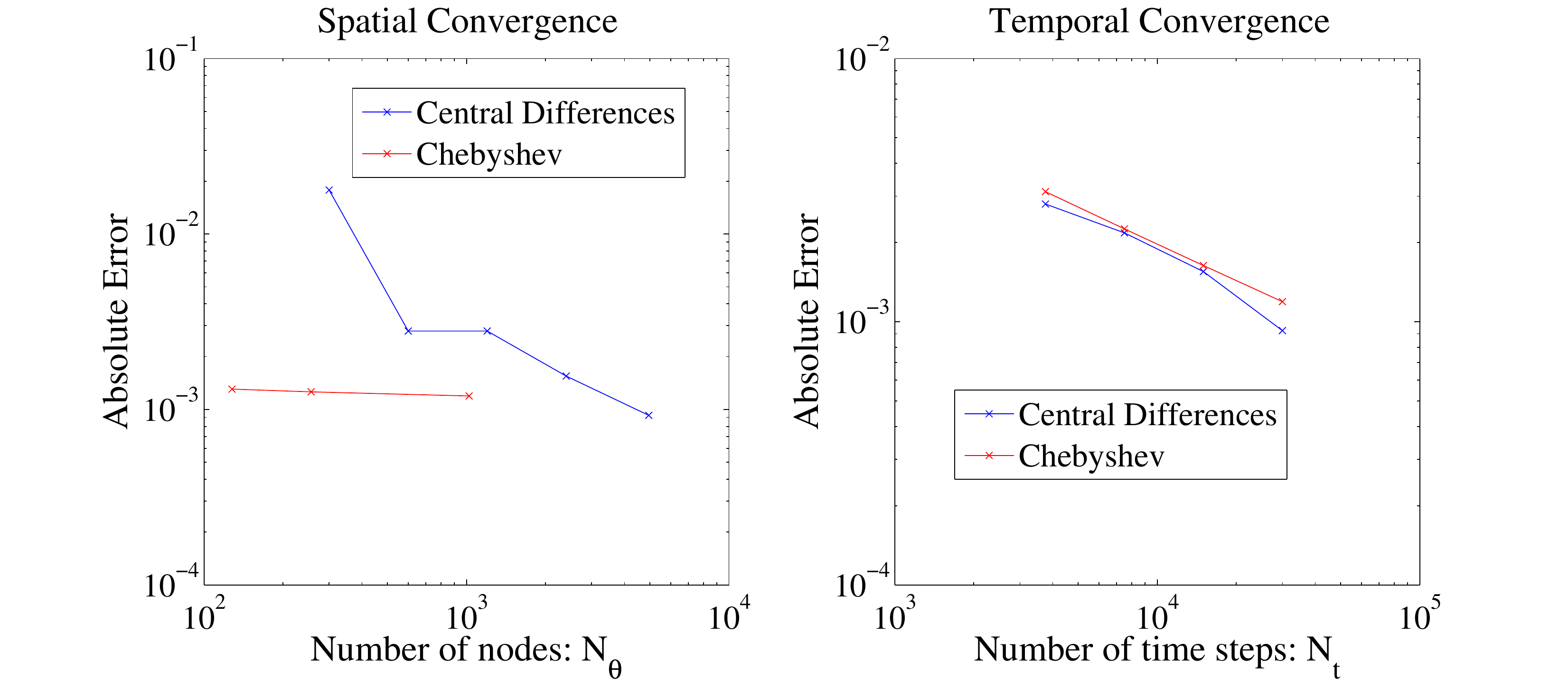}
\caption[Error Convergence of the Stationary State of the Buying Frontier]{\label{Ch3v1estacconv} Spatial (left) and Temporal (right) error convergence, in logarithmic scale, of the Stationary State of the Buying Frontier for the CD(blue) and Chebyshev (red) methods.}
\end{figure}

In this experiment, temporal error is dominant compared with respect to the spatial error in the Chebyshev method. In the case of the CD method, the error depends more in both the spatial and temporal discretizations. On the left side picture, we can see that the Chebyshev method reaches the error marked by the time discretization with the smallest number of nodes Therefore, if a high precision is required, Chebyshev will perform better than the Central Differences method.

The error behaviour of the Selling Frontier is similar to the one of the Buying Frontier.

\subsection{Performance Analysis}\label{Ch3numresultsperformance}

In this section we compare the relative performance
of the pseudospectral and finite difference methods. First of all, we fix several time and spatial discretization parameters:
\begin{enumerate}[(i)]
\item $\Delta t\in[0.02, \ 3\ldotp9^{-4}]$
\item $N_{\theta}\in[141, \ 1024]$ (Chebyshev)
\item $N_{\theta}\in[300, \ 6000]$ (Central Differences)
\end{enumerate}
and solve the problem with all the combinations of the different discretizations for both methods.

The lower and upper bounds of $N_{\theta}$ in the Central Differences method can be taken smaller or bigger. The criteria that we have employed is such that the numerical error varies between $10^{-4}$ and $10^{-8}$. The same reads for the upper bound of $N_{\theta}$ in the Chebyshev method.

We point that during the implementation of the method, we observed that if $N_{\theta}$ was not big enough, the location of the frontiers may oscillate (due to the Gibbs effect or to the fact that the polynomials are not accurate enough), complicating the location of $BR^{\textbf{N}}_F$ and $SR^{\textbf{N}}_F$ in (\ref{frontiers}). The Chebyshev spectral method is effective once enough resolution has been reached. This behaviour is typical of high order methods, see \cite{Frutos}. The employment of the adaptive interval $I(t_l)$ and an enough amount of interpolation nodes avoids the oscillations and allows to obtain just one numerical approximation of $BR^{\textbf{N}}_F$ and $SR^{\textbf{N}}_F$ in (\ref{frontiers}). The oscillations may appear if the following (empirical) bounds are violated
\begin{equation}\label{Ch3empirconstr}
\Delta t > 0.1, \quad  \Delta t < \frac{C}{N_{\theta}^{C_1}},
\end{equation}
where $C_1 \geq 1$ and numerical experiments suggest that $C_1$ might be a growing function of $N_{\theta}$.

The lowest value of $N_{\theta}$ in the Chebyshev method was chosen so that no oscillations appear. If a smaller number of interpolation nodes is chosen, the solution oscillates and the error worsens.

We plot the value of $\text{RMSE}_{\{N_{\theta},N_t\}}\left(v^{\textbf{N}}(0,t)\right)$ (\ref{Ch3errrorcuadmed}) versus the computational time employed in computing $v^{\textbf{N}}$ for each different spatial and temporal meshes in logarithmic scale.

\begin{figure}[h]
\centering
\includegraphics[width=13.5cm,height=6cm]{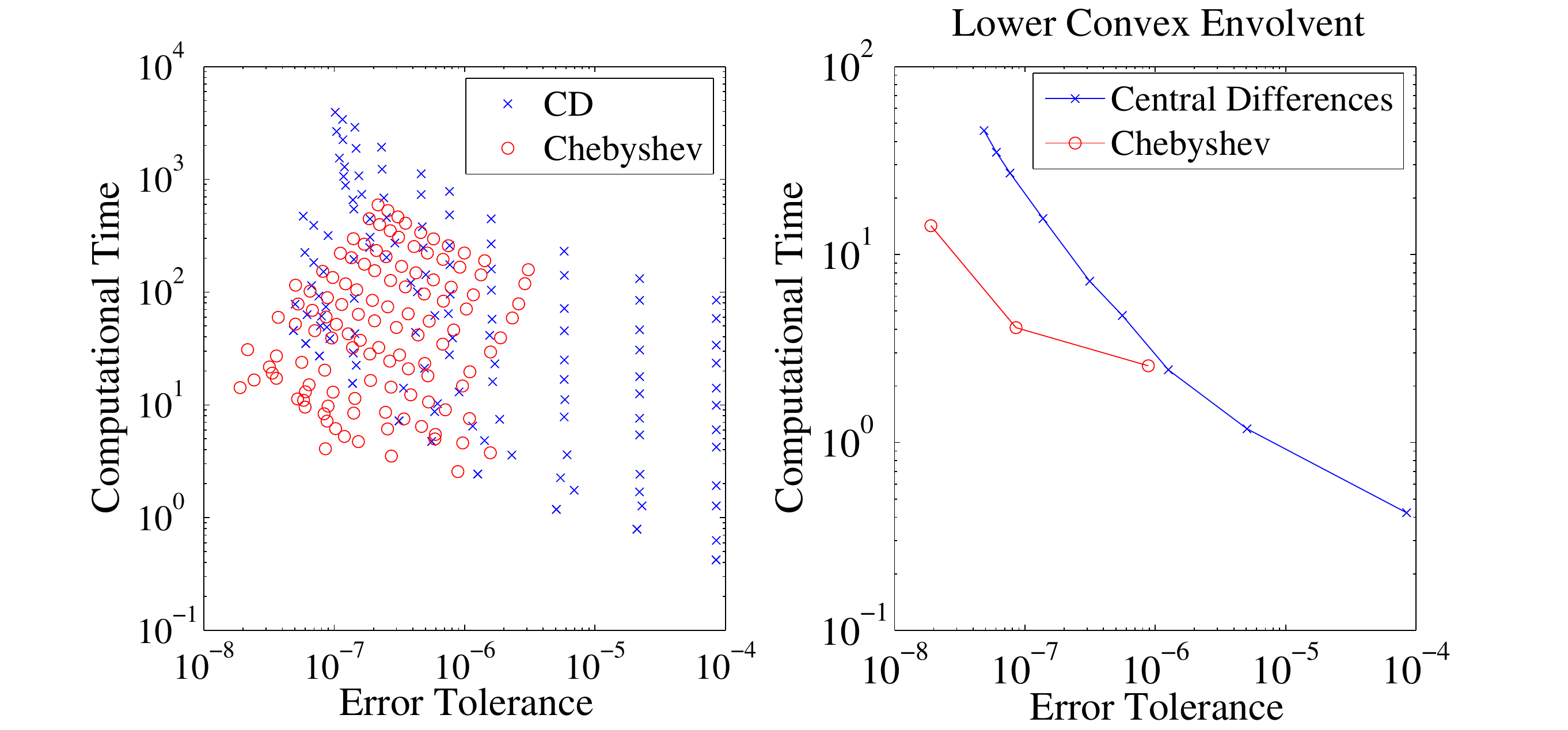}
\caption[{Performance comparison of the Error at $v^{\textbf{N}}(0,t)$.}]{\label{Ch3v1performanceerror} Performance comparison of the Error at $v^{\textbf{N}}(0,t)$.  In logarithmic scale, we plot (left), the value of RMSE versus the total computational costs of CD (blue) and Chebyshev (red) methods and their respective lower enveloping curves (right).}
\end{figure}

The left-side picture of Figure \ref{Ch3v1performanceerror} represents the cloud of results for the different discretizations of each method. The right-side, which is more visual, represents the lower convex enveloping curve.

With the right-side picture, we can obtain an approximate behaviour of the evolution of the error versus the required computational time to reach that precision. We fix the error tolerance that we require for our problem and find which method and spatial and time discretization reaches it first.

As we can see, the CD method (blue in Figure \ref{Ch3v1performanceerror}) performs better if we do not require a high precision. If a higher precision is required, Chebyshev (red in Figure \ref{Ch3v1performanceerror}) performs better than CD.

A similar behaviour can be observed if we compare the errors of the rest of cases where we have explicit formulas.

\section{Conclusions}\label{Ch3Conclus}

The homothetic property of the Potential Utility function has been used to restate the investment problem in polar coordinates. This has allowed us to give an equivalent formulation of the problem in a bounded spatial domain.

Although some of the numerical difficulties that appear with the parabolic double obstacle problem are avoided, other problems may appear if we employ spectral methods. The Gibbs effect, which comes from the fact that the objective function is continuous but not differentiable at maturity, can complicate the location of the frontiers, but this issue can be circumvented by the employment of a time-adapted spatial mesh (Subsection \ref{Ch3MACCMAM}).

Simpler methods, as Central Differences, are not affected by the Gibbs effect and they are easier to implement. Nevertheless, they  require more computational work if high precisions are needed.

Further work may include the extension of the model including a consumption term or the design of spectral methods to optimal investment problems with other Utility functions, like the Exponential Utility. Furthermore, through the Indifference Pricing technique (see \cite{Carmona} and \cite{Davis2}), these kind of models can be applied to option valuation.


\begin{thebibliography}{99}

\bibitem{Achdou} Achdou Y., Pironneau O. (2007): Finite Element Method for Option Pricing. Universit\'e Pierre et Marie Curie.

\bibitem{Arregui} Arregui I., V\'azquez C., {\em Numerical solution of an optimal investment problem with proportional transaction costs\/}, Journal of Computational and Applied Mathematics, 236 (2012), 2923-2937.

\bibitem{Benameur} Ben-Ameur H., J. de Frutos, T. Fakhfakh and Diaby V., {\em Upper and Lower Bounds for Convex Value Functions of Derivative Contracts\/}, Economic Modelling, 34 (2013), 69-75.

\bibitem{Breton} Breton, M. and de Frutos, J. , {\em Option Pricing under GARCH Processes by PDE Methods\/},  Operations Research, 58 (2010), 1148-1157.

\bibitem{Breton2} Breton, M. and de Frutos, J., {\em Approximation of Dynamic Programs},  in Handbook of Computational Finance, 633-649, Jin-Chuan Duan, James E. Gentle, and Wolfgang H\"{a}rdle(eds), Springer, 2012.

\bibitem{Canuto} C. Canuto, M.Y. Hussaini, A. Quarteroni and T.A. Zang, {\em Spectral Methods. Fundamentals in single domains\/}, Springer, Berlin, 2006.

\bibitem{Carmona} R. Carmona, {\em Indifference Pricing\/}, Princeton University Press, Princeton, 2009.

\bibitem{Chiarella} Chiarella, C., El-Hassan, N. and A. Kucera, A {\em Evaluation of American option prices in a path integral
framework using Fourier-Hermite series expansion\/} , Journal
of Economic Dynamics and Control 23 (1999), 1387-1424.

\bibitem{Cvitanic} Cvitani\'{c} J., Karatzas I.,  {\em Hedging and Portfolio Optimization under Transaction Costs: A Martingale Approach\/}, Mathematical Finance, 6 (1996), 113-165.

\bibitem{Davis} Davis M.H.A., Norman A.R., {\em Portfolio selection with transaction costs\/}, Mathematics of Operations Research, 15 (1990), 676-713.

\bibitem{Davis2} Davis M.H.A., Panas V.G., Zariphopoulou T. {\em European Option Pricing with transaction costs\/}, SIAM Journal of Control and Optimization, 31 (1993), 470-493.

\bibitem{DaiYi} Day M., Yi F., {\em Finite-Horizon Optimal Investment with Transaction Costs: A Parabolic Double Obstacle Problem\/}, Journal of Differential Equations, 246 (2009), 1445-1469.

\bibitem{DuanHandbook} Duan J.C., Gentle J.E. and H\"{a}rdle W.(eds) {\em Handbook of Computational Finance\/}, Springer, 2012.

\bibitem{Duan2} Duan J.C., Simonato (1998): Empirical Martingale Simulation. Management Science, 44, 1218-1233.

\bibitem{Frutos} de Frutos, J., {\em A Spectral Method for bonds\/}, Computers and Operations Research, 35 (2008), 64-75.

\bibitem{Gaton} Gat\'on, V., {\em Cuatro ensayos sobre valoraci\'on de derivados y estrategias de inversi\'on\/}, Ph.D. thesis, University of Valladolid, Valladolid, 2016.



\bibitem{Lyuu} Lyuu Y.,Wu C. (2005): On accurate and Provably Efficient GARCH Option Pricing Algorithms. Quantitative Finance, 2, 181-198.


\bibitem{Magill} Magill M.J.P., Constatinides G.M., {\em Portfolio selection with transaction costs\/}, Journal of Economic Theory, 13 (1976), 245-263.

\bibitem{Merton} Merton R.C., {\em Optimal consumption and portfolio rules in a continuous time model\/}, Journal of Economic Theory, 3 (1971), 373-413.

\bibitem{Oosterlee} Zhang, B. and Oosterlee, C. W.{\em Pricing of early-exercise Asian options under L\'{e}vy processes based on Fourier cosine expansions\/},
Appl. Numer. Math., 78 (2014), 14-30.

\bibitem{Ritchen} Ritchen P, Trevor R. (1999): Pricing Options under Generalized GARCH and Stochastic Volatility Processes. The Journal of Finance, 54, 377-402.


\bibitem{Shreve} Shreve S.E., Soner H.M., {\em Optimal investment and consumption with transaction costs\/}, Annals of Applied Probablity, 4  (1994), 609-692.

\bibitem{Stentoft} Stentof L. (2004): Pricing American Options when the underlying asset follows GARCH processes. Journal of Empiriccal Ginance, 12, 576-611.



\end{thebibliography}
\end{document}